\documentclass[sigconf]{acmart}
\AtBeginDocument{%
  \providecommand\BibTeX{{%
    \normalfont B\kern-0.5em{\scshape i\kern-0.25em b}\kern-0.8em\TeX}}}

\setcopyright{acmcopyright}
\copyrightyear{2024}
\acmYear{2024}
\setcopyright{acmlicensed}\acmConference[WWW '24]{Proceedings of the ACM Web Conference 2024}{May 13--17, 2024}{Singapore, Singapore}
\acmBooktitle{Proceedings of the ACM Web Conference 2024 (WWW '24), May 13--17, 2024, Singapore, Singapore}
\acmDOI{10.1145/3589334.3645614}
\acmISBN{979-8-4007-0171-9/24/05}

\usepackage{subfigure}

\newtheorem{lemma}{Lemma}
\usepackage{multirow}
\usepackage{multicol}
\newtheorem{problem definition}{Problem Definition}

\usepackage{algorithm}
\usepackage{algorithmic}
\usepackage{enumitem}
\usepackage{balance}

\begin{document}

%%
%% The "title" command has an optional parameter,
%% allowing the author to define a "short title" to be used in page headers.
\title{High-Frequency-aware Hierarchical Contrastive Selective Coding for Representation Learning on Text-attributed Graphs}

\author{Peiyan Zhang}
\affiliation{\institution{Hong Kong University of \\ Science and Technology}\country{Hong Kong}}
\email{pzhangao@cse.ust.hk}

\author{Chaozhuo Li}
\authornote{Chaozhuo Li is the corresponding author}
\affiliation{\institution{Microsoft Research Asia}\city{Beijing}\country{China}}
\email{cli@microsoft.com}

\author{Liying Kang}
\affiliation{\institution{Hong Kong Polytechnic University}\country{Hong Kong}}
\email{lykangc12@gmail.com}

\author{Feiran Huang}
\affiliation{\institution{Jinan University}\country{China}}
\email{huangfr@jnu.edu.cn}

\author{Senzhang Wang}
\affiliation{\institution{Central South University}\country{China}}
\email{szwang@csu.edu.cn}

\author{Xing Xie}
\affiliation{\institution{Microsoft Research Asia}\city{Beijing}\country{China}}
\email{xing.xie@microsoft.com}

\author{Sunghun Kim}
\affiliation{\institution{Hong Kong University of \\ Science and Technology}\country{Hong Kong}}
\email{hunkim@cse.ust.hk}

\renewcommand{\shortauthors}{Peiyan Zhang, et al.}

%%
%% The abstract is a short summary of the work to be presented in the
%% article.
\begin{abstract}
    We investigate node representation learning on text-attributed graphs (TAGs), where nodes are associated with text information. Although recent studies on graph neural networks (GNNs) and pretrained language models (PLMs) have exhibited their power in encoding network and text signals, respectively, less attention has been paid to delicately coupling these two types of models on TAGs. Specifically, existing GNNs rarely model text in each node in a contextualized way; existing PLMs can hardly be applied to characterize graph structures due to their sequence architecture. To address these challenges, we propose HASH-CODE, a \textbf{H}igh-frequency \textbf{A}ware \textbf{S}pectral \textbf{H}ierarchical \textbf{Co}ntrastive Selective Co\textbf{d}ing method that integrates GNNs and PLMs into a unified model. Different from previous “cascaded architectures” that directly add GNN layers upon a PLM, our HASH-CODE relies on five self-supervised optimization objectives to facilitate thorough mutual enhancement between network and text signals in diverse granularities. Moreover, we show that existing contrastive objective learns the low-frequency component of the augmentation graph and propose a high-frequency component (HFC)-aware contrastive learning objective that makes the learned embeddings more distinctive. Extensive experiments on six real-world benchmarks substantiate the efficacy of our proposed approach. In addition, theoretical analysis and item embedding visualization provide insights into our model interoperability.  
\end{abstract}

%%
%% The code below is generated by the tool at http://dl.acm.org/ccs.cfm.
%% Please copy and paste the code instead of the example below.
%%
% \begin{CCSXML}
% <ccs2012>
%  <concept>
%   <concept_id>00000000.0000000.0000000</concept_id>
%   <concept_desc>Do Not Use This Code, Generate the Correct Terms for Your Paper</concept_desc>
%   <concept_significance>500</concept_significance>
%  </concept>
%  <concept>
%   <concept_id>00000000.00000000.00000000</concept_id>
%   <concept_desc>Do Not Use This Code, Generate the Correct Terms for Your Paper</concept_desc>
%   <concept_significance>300</concept_significance>
%  </concept>
%  <concept>
%   <concept_id>00000000.00000000.00000000</concept_id>
%   <concept_desc>Do Not Use This Code, Generate the Correct Terms for Your Paper</concept_desc>
%   <concept_significance>100</concept_significance>
%  </concept>
%  <concept>
%   <concept_id>00000000.00000000.00000000</concept_id>
%   <concept_desc>Do Not Use This Code, Generate the Correct Terms for Your Paper</concept_desc>
%   <concept_significance>100</concept_significance>
%  </concept>
% </ccs2012>
% \end{CCSXML}

% \ccsdesc[500]{Do Not Use This Code~Generate the Correct Terms for Your Paper}
% \ccsdesc[300]{Do Not Use This Code~Generate the Correct Terms for Your Paper}
% \ccsdesc{Do Not Use This Code~Generate the Correct Terms for Your Paper}
% \ccsdesc[100]{Do Not Use This Code~Generate the Correct Terms for Your Paper}

%%
%% Keywords. The author(s) should pick words that accurately describe
%% the work being presented. Separate the keywords with commas.

%%
%% This command processes the author and affiliation and title
%% information and builds the first part of the formatted document.
\maketitle

\section{Introduction}

Graphs are pervasive in the real world, and it is common for nodes within these graphs to be enriched with textual attributes, thereby giving rise to text-attributed graphs (TAGs)~\citep{zhao2022learning}. For instance, academic graphs~\citep{tang2008arnetminer} incorporate  papers replete with their titles and abstracts, whereas social media networks~\citep{zhang2016geoburst} encompass tweets accompanied by their textual content. Consequently, the pursuit of learning within the realm of TAGs has assumed significant prominence as a research topic spanning various domains, \textit{e.g.,} network analysis~\citep{wang2019heterogeneous}, recommender systems~\citep{zhang2019heterogeneous}, and anomaly detection~\citep{liu2019fine}.

In essence, graph topology and node attributes comprise two integral components of TAGs. 
% Graph topology encapsulates the local structures of nodes by delineating their interconnections within a graph, while node attributes convey the semantics of nodes by endowing them with textual features. 
Consequently, the crux of representation learning on TAGs lies in the amalgamation of graph topology and node attributes. Previous works mainly adopt a cascaded architecture~\citep{jin2021bite,li2021adsgnn,zhang2019shne,zhu2021textgnn} (Figure~\ref{fig:example}(a)), which entails encoding the textual attributes of each node with Pre-trained Language Models (PLMs), subsequently utilizing the PLM embeddings as features to train a Graph Neural Network (GNN) for message propagation~\citep{gururangan2020don,chien2021node,yasunaga2022linkbert}. However, as the modeling of node attributes and graph topology are segregated, this learning paradigm harbors conspicuous limitations. Firstly, the link connecting two nodes is not
utilized when generating their text representations. In fact, linked
nodes can benefit each other regarding text semantics understanding. For example, given a paper on "LDA" and its citation nodes which are related to topic modeling, the "LDA" can be more likely interpreted as "Latent Dirichlet Allocation" rather than "Linear Discriminant Analysis". In addition, this paradigm may yield textual embeddings that are not pertinent to downstream tasks, thereby impeding the model's ability to learn node representations suitable for such tasks. Moreover, given that the formation of the graph's topological structure is intrinsically driven by the node attribute~\citep{zhao2022learning}, this paradigm may adversely affect the comprehension of the graph topology.

\begin{figure}[t]
    \centering
    \includegraphics[width=\linewidth]{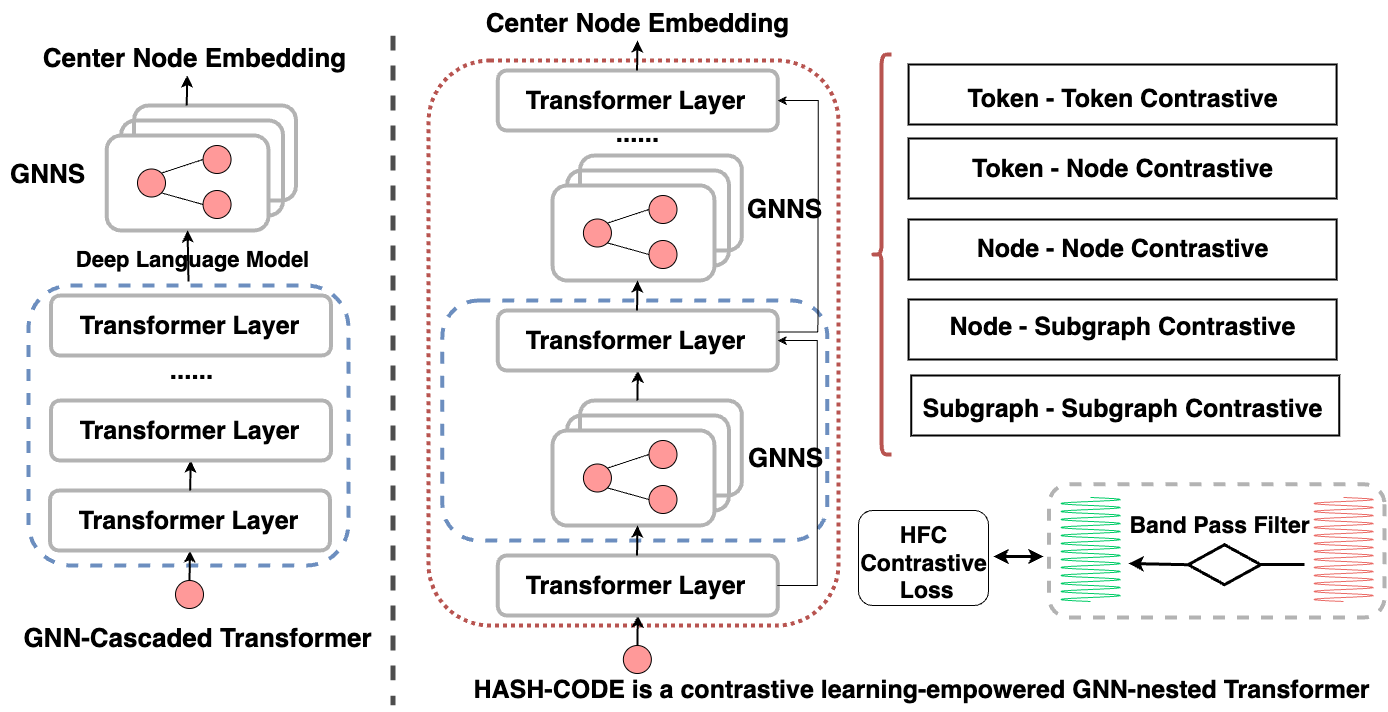}
    \caption{(a) An illustration of GNN-cascaded transformer. (b) An illustration of our proposed contrastive learning-empowered GNN-nested transformer. The red and green twines denote the original graph signals and the mixed LFC and HFC signals from the spectral perspective.}
    \label{fig:example}
    % \vspace{-0.3cm}
\end{figure}

Fortunately, recent efforts have been undertaken~\citep{li2017ppne,yang2021graphformers,bi2021leveraging,pang2022improving,jin2022heterformer,yan2024comprehensive} to co-train GNNs and LMs within a unified learning framework. For example, GraphFormers~\citep{yang2021graphformers} introduces GNN-nested transformers, facilitating the joint encoding of text and node features. Heterformer~\citep{jin2022heterformer} alternately stacks the graph aggregation module and a transformer-based text encoding module into a cohesive model to capture network heterogeneity. Despite the demonstrated efficacy of existing methods, they are encumbered by two primary drawbacks that may undermine the quality of representation learning. Firstly, these methods typically employ supervised training, necessitating a substantial volume of labeled data. However, in numerous scientific domains, labeled data are scarce and expensive to obtain~\citep{hu2019strategies,wang2021self}. 
% For instance, accurately labeling an unknown gene often demands an extensive understanding of molecular biology, presenting a formidable challenge even for seasoned researchers~\citep{hu2019strategies}. 
Secondly, these methods rely exclusively on limited optimization objectives to learn the entire model. When GNNs and LMs are jointly trained, the associated parameters are also learned through the constrained optimization objectives. It has been observed that such an optimization approach fails to capture the fine-grained correlations between textual features and graphic patterns~\citep{yang2021graphformers, zhou2020s3}. Consequently, the importance of learning graph representations in an unsupervised or self-supervised manner is becoming increasingly paramount.

In order to tackle the aforementioned challenges, we draw inspiration from the concept of contrastive learning to enhance representation learning on TAGs. Contrastive learning~\citep{chen2020simple,hassani2020contrastive,he2020momentum,velickovic2019deep} refines representations by drawing positive pairs closer while maintaining a distance between negative pairs. As data sparsity and limited supervision signals constitute the two primary learning obstacles associated with existing co-training methods, contrastive learning appears to offer a promising solution to both issues: it capitalizes on intrinsic data correlations to devise auxiliary training objectives and bolsters data representations with an abundance of self-supervised signals.

% To address the above challenges, we borrow the idea of contrastive learning for improving the representation learning on TAGs. Contrastive learning~\citep{chen2020simple,hassani2020contrastive,he2020momentum,velickovic2019deep} is a newly emerging paradigm, which learns representations by pushing positive pairs closer while keeping negative pairs far apart. As previously discussed, data sparsity and limited supervision signals are the two major learning issues with the existing co-training methods. Fortunately, contrastive learning seems to provide a promising solution to both problems: it utilizes the intrinsic data correlations to devise auxiliary training objectives and enhances the data representations with rich self-supervised signals. 

In practice, representation learning on TAGs with contrastive learning is non-trivial, primarily encountering the following three challenges: (1)~\textit{How to devise a learning framework that capitalizes on the distinctive data properties of TAGs?} The contextual information within TAGs is manifested in a multitude of forms or varying intrinsic characteristics, such as tokens, nodes, or sub-graphs, which inherently exhibit complex hierarchical structures. Moreover, these hierarchies are interdependent and exert influence upon one another. How to capitalize these unique properties of TAGs remains an open question. (2)~\textit{How to design effective contrastive tasks?} To obtain an effective node embedding that fully encapsulates the semantics, relying solely on the hierarchical topological views of TAGs remains insufficient. Within TAGs, graph topological views and textual semantic views possess the capacity to mutually reinforce one another, indicating the importance of exploring the cross-view contrastive mechanism. Moreover, the hierarchies in TAGs can offer valuable guidance in selecting positive pairs with analogous semantics and negative pairs with divergent semantics, an aspect that has received limited attention in existing research~\citep{xu2021self,kim2021self}. (3)~\textit{How to learn distinctive representations?} In developing the contrastive learning framework, we draw inspiration from the recently proposed spectral contrastive learning method~\citep{haochen2021provable}, which outperforms several contrastive baselines with solid theoretical guarantees. However, we demonstrate that, from a spectral perspective, the spectral contrastive loss primarily learns the low-frequency component (LFC) of the graph, significantly attenuating the effects of high-frequency components (HFC). Recent studies suggest that the LFC does not necessarily encompass the most vital information~\citep{bo2021beyond,chen2019drop}, and would ultimately contribute to the over-smoothing problem~\citep{cai2020note,chen2020measuring,li2018deeper,liu2020towards}, causing node representations to converge to similar values and impeding their differentiation. Consequently, more explorations are needed to determine how to incorporate the HFC to learn more discriminative embeddings.

To this end, we present a novel \textbf{H}igh-frequency \textbf{A}ware \textbf{S}pectral \textbf{H}ierarchical \textbf{Co}ntrastive Selective Co\textbf{d}ing framework (\textbf{HASH-CODE}) to enhance TAG representation learning. Building upon a GNN and Transformer architecture~\citep{yang2021graphformers,zhu2021textgnn}, we propose to jointly train the GNN and Transformer with self-supervised signals (Figure~\ref{fig:example}(b) depicts this architecture). The primary innovation lies in the contrastive joint-training stage. Specifically, we devise five self-supervised optimization objectives to capture hierarchical intrinsic data correlations within TAGs. These optimization objectives are developed within a unified framework of contrastive learning. Moreover, we propose a loss that can be succinctly expressed as a contrastive learning objective, accompanied by robust theoretical guarantees. Minimizing this objective results in more distinctive embeddings that strike a balance between LFC and HFC. Consequently, the proposed method is capable of characterizing correlations across varying levels of granularity or between different forms in a general manner. 

Our main contributions are summarized as follows:
\begin{itemize}[leftmargin=*,noitemsep,topsep=0pt]
    \item We propose five self-supervised optimization
objectives to maximize the mutual information of context information in different forms or granularities.
    \item We systematically examine the fundamental limitations of spectral contrastive loss from the perspective of spectral domain. We prove that it learns the LFC and propose an HFC-aware contrastive learning objective that makes the learned embeddings more discriminative.
    \item Extensive experiments conducted on three million-scale text-attributed graph datasets demonstrate the effectiveness of our proposed approach.
\end{itemize}

% The rest of this article is structured in the following manner: Section~\ref{sec:related} delineates the current methodologies employed for representation learning on TAGs, as well as recent advancements in the realm of contrastive learning. Section~\ref{sec:bg} furnishes the necessary background information and a comprehensive theoretical examination of the contrastive learning loss. Subsequently, the proposed approach is expounded upon in Section~\ref{sec:method}. In Section~\ref{sec:exp}, comparative experiments and ablation analyses are executed to substantiate the efficacy of our method. Lastly, Section~\ref{sec:con} encapsulates the conclusions drawn from this research.

% The rest of this article is organized as follows: Section~\ref{sec:related} introduces the existing methods for the representation learning on TAGs and recent advances in contrastive learning. In Section~\ref{sec:bg}, We provide the background and our theoretical analysis on the contrastive learning loss. Then the proposed method is presented in Section~\ref{sec:method}. In Section~\ref{sec:exp}, we conduct comparison experiments and ablation analysis to verify our method. Finally, Section~\ref{sec:con} provides our conclusions.
% \vspace{-0.6cm}

\section{Related Work}
\label{sec:related}
\subsection{Representation Learning on TAGs}

Representation learning on TAGs constitutes a significant research area across multiple domains, including natural language processing~\cite{wang2016linked,wang2016text,yang2022reinforcement,jin2022towards}, information retrieval~\cite{wang2019improving,xu2019deep,jin2023predicting,jin2022code}, and graph learning~\cite{yang2015network,yasunaga2017graph,long2021hgk,long2020graph}. In order to attain high-quality representations for TAGs, it is imperative to concurrently harness techniques from both natural language understanding and graph representation. The recent advancements in pretrained language models (PLM) and graph neural networks (GNN) have catalyzed the progression of pertinent methodologies.

\textbf{Seperated Training.} 
A number of recent efforts strive to amalgamate GNNs and LMs, thereby capitalizing on the strengths inherent in both models. The majority of prior investigations on TAGs employ a "cascaded architecture"~\citep{jin2021bite,li2021adsgnn,zhang2019shne,zhu2021textgnn}, in which the text information of each node is initially encoded through transformers, followed by the aggregation of node representations via GNNs. Nevertheless, these PLM embeddings remain non-trainable during the GNN training phase. Consequently, the model performance is adversely impacted by the semantic modeling process, which bears no relevance to the task and topology at hand.

% Some recent endeavors aim at the integration of GNNs and LMs, which enables one to benefit from the strengths of both models. Most previous studies on TAGs adopt a “cascaded architecture"~\citep{jin2021bite,li2021adsgnn,zhang2019shne,zhu2021textgnn}, where the text information of each node is first encoded via transformers, then the node representations are aggregated via GNNs. However, these PLM embeddings are still non-trainable during the GNN training phase. Therefore, the model performance suffers from the semantic modeling process that is irrelevant to the task and topology. 

\textbf{Co-training.} 
In an attempt to surmount these challenges, concerted efforts have been directed towards the co-training of GNNs and PLMs within a unified learning framework. GraphFormers~\cite{yang2021graphformers} presents GNN-nested transformers, facilitating the concurrent encoding of text and node features. Heterformer~\cite{jin2022heterformer} alternates between stacking the graph aggregation module and a transformer-based text encoding module within a unified model, thereby capturing network heterogeneity. However, these approaches solely depend on a single optimization objective for learning the entire model, which considerably constrains their capacity to discern the fine-grained correlations between textual and graphical patterns. 

\subsection{Contrastive Learning}
\textbf{Empirical Works on Contrastive learning.} 
Contrastive methods~\citep{chen2020simple,chen2020improved,he2020momentum} derive representations from disparate views or augmentations of inputs and minimize the InfoNCE loss~\citep{oord2018representation}, wherein two views of identical data are drawn together, while views from distinct data are repelled. The acquired representation can be utilized to address a wide array of downstream tasks with exceptional performance. In the context of node representation learning on graphs, DGI~\cite{velickovic2019deep} constructs local patches and global summaries as positive pairs. GMI~\cite{peng2020graph} is designed to establish a contrast between the central node and its local patch, derived from both node features and topological structure. MVGRL~\cite{hassani2020contrastive} employs contrast across views and explores composition between varying views.

% Contrastive methods~\citep{chen2020simple,chen2020improved,he2020momentum} learn representations from different views or augmentations of inputs and minimize the InfoNCE loss~\citep{oord2018representation}, where two views of the same data are attracted while views from different data are repulsed. The learned representation can be applied to address a broad spectrum of downstream tasks with high performance. For the node representation learning on graphs, DGI~\cite{velickovic2019deep} builds local patches and global summary as positive pairs. GMI~\cite{peng2020graph} is proposed to contrast between the center node and its local patch from node features and topological structure. MVGRL~\cite{hassani2020contrastive} employs contrast across views and experiments composition between different views. 

\textbf{Theoretical works on Contrastive Learning.}
The exceptional performance exhibited by contrastive learning has spurred a series of theoretical investigations into the contrastive loss. The majority of these studies treat the model class as a black box, with notable exceptions being the work of~\cite{lee2021predicting}, which scrutinizes the learned representation with linear models, and the research conducted by~\cite{tian2022deep} and~\cite{wen2021toward}, which examine the training dynamics of contrastive learning for linear and 2-layer ReLU networks. Most relevant to our research is the study by~\cite{saunshi2022understanding}, which adopts a spectral graph perspective to analyze contrastive learning methods and introduces the spectral contrastive loss. We ascertain that the spectral contrastive loss solely learns the LFC of the graph.

% The remarkable performance of contrastive learning has stimulated a series of theoretical works that investigate the contrastive loss, most of which regard the model class as a black box except for the work of Lee et al. (2020) which examines the learned representation with linear models, and the works of Tian (2022) and Wen and Li (2021) which study the training dynamics of contrastive learning for linear and 2-layer ReLU networks. Most pertinent to our work is Saunshi et al. (2022) where a spectral graph perspective is adopted to analyze the contrastive learning methods and proposes the spectral contrastive loss. We discover that the spectral contrastive loss only learns the LFC of the graph.
% A spectral graph point of view is also taken in~\citep{haochen2021provable} to analyze contrastive learning methods and proposes the spectral contrastive loss. However, spectral contrastive loss only learns the LFC of the graph.

Different from the existing works, our research represents the first attempt to contemplate the correlations inherent within the contextual information as self-supervised signals in TAGs. We endeavor to maximize the mutual information among the views of the token, node, and subgraph, which encompass varying levels of granularity within the contextual information. Our HFC-aware loss facilitates the learning of more discriminative data representations, thereby enhancing the performance of downstream tasks.

% Different from the above approaches, our work is the first to consider the correlations within the contextual information as the self-supervised signals in TAGs. We maximize the mutual information among the views of the token, node, and subgraph, which are in different levels of granularity of the contextual information. Our HFC-aware loss can learn more discriminative data representations that improve downstream tasks' performance.

% Here we mainly focus on reviewing the graph-related contrastive learning methods. Specifically, DGI~\cite{velickovic2019deep} builds local patches and global summary as positive pairs, and utilizes Infomax~\cite{linsker1988self} theory to contrast. Along this line, GMI~\cite{peng2020graph} is proposed to contrast between the center node and its local patch from node features and topological structure. MVGRL~\cite{hassani2020contrastive} employs contrast across views and experiments composition between different views. 
\section{Preliminaries}
\label{sec:bg}
In this section, we first give the definition of the text-attributed graphs (TAGs) and formulate the node representation learning problem on TAGs. Then, we introduce our proposed HFC-aware spectral contrastive loss.

% In this section, we first give the definition of the text -attributed graphs (TAGs) and formulate the node representation learning problem on TAGs. Next, we revisit the spectral clustering and spectral contrastive loss. Finally, we introduce our proposed HFC-aware loss.

\subsection{Definition (Text-attributed Graphs)}
% \textbf{Definition (Text -attributed Graphs).} 
A text-attributed graph is defined as $\mathcal{G}=(\mathcal{V},\mathcal{E})$, where $\mathcal{V}=\{v_{1},...,v_{N}\}$ and $\mathcal{E}$ denote the set of nodes and edges, respectively. Let $A\in \mathbb{R}^{N\times N}$ be the adjacency matrix of the graph such that $A_{i,j}=1$ if $v_{j}\in \mathcal{N}(v_{i})$, otherwise $A_{i,j}=0$. Here $\mathcal{N}(.)$ denotes the one-hop neighbor set of a node. Besides, each node $v_{i}$ is associated with text information. 

% $\mathcal{V}$, $\mathcal{E}$, $\mathcal{X}$ represent the sets of nodes, edges and textual attributes, respectively. Each node $v\in \mathcal{V}$ is associated with text information $x\in \mathcal{X}$. 

% \begin{table}[t]
%     \centering
%     \caption{Main Notations Used in This Article.}
%     {
%     \begin{tabular}{cc}
%         \toprule
%            Notation      & Description \\
%         \midrule
%         $\mathcal{G}=(\mathcal{V},\mathcal{E})$  & 13,647,591\\
%         \#Items & 5,643,688\\
%         \#N   & 4.71 \\
%         \#Train  & 22,146,934 \\
%         \#Valid  & 30,000 \\
%         \#Test   & 306,742 \\
%         \bottomrule
%     \end{tabular}}
%     \label{tab:notation}
%     % \vspace{-0.4cm}
% \end{table}

\subsection{Problem Statement}
Given a textual attibuted graph $\mathcal{G}=(\mathcal{V},\mathcal{E})$, the task is to build a model $f_{\theta}: \mathcal{V}\rightarrow \mathbb{R}^{K}$ with parameters $\theta$ to learn the node embedding matrix $F\in \mathbb{R}^{N\times K}$, taking network structures and text semantics into consideration, where $K$ denotes the number of feature channels. The learned embedding matrix $F$ can be further utilized in downstream tasks, \textit{e.g.,} link prediction, node classification, \textit{etc.}

\subsection{HFC-aware Spectral Contrastive Loss}

\par An important technique in our approach is the high-frequency aware spectral contrastive loss. It is developed based on the analysis of the conventional spectral contrastive loss~\citep{haochen2021provable}. Given a node $v$, the conventional spectral contrastive loss is defined as:
\begin{equation}
\begin{aligned}
    \mathcal{L}_{Spectral}(v,v^{+},v^{-}) &= -2\cdot \mathbb{E}_{v,v^{+}}[f_{\theta}(v)^{T}f_{\theta}(v^{+})]\\
    &\ +\mathbb{E}_{v,v^{-}}[(f_{\theta}(v)^{T}f_{\theta}(v^{-}))^{2}],
\end{aligned}
\end{equation}
where $(v,v^{+})$ is a pair of positive views of node $v$, $(v,v^{-})$ is a pair of negative views, and $f_{\theta}$ is a parameterized function from the node to $\mathbb{R}^{K}$. Minimizing $\mathcal{L}_{Spectral}$ is equivalent to spectral clustering on the population view graph~\citep{haochen2021provable}, where the top smallest eigenvectors of the Laplacian matrix are preserved as the columns of the final embedding matrix $F$. 

\par In Appendix~\ref{sec:spectral}, we demonstrate that, from a spectral perspective, $\mathcal{L}_{Spectral}$ primarily learns the low-frequency component (LFC) of the graph, significantly attenuating the effects of high-frequency components (HFC). Recent studies suggest that the LFC does not necessarily encompass the most vital information~\citep{bo2021beyond,chen2019drop}, and would ultimately contribute to the over-smoothing problem~\citep{cai2020note,chen2020measuring,li2018deeper,liu2020towards}.

As an alternative of such low-pass filter, to introduce HFC, we propose our HFC-aware spectral contrastive loss as follows: 
\begin{equation}
\label{eq:hfc}
\begin{aligned}
    \mathcal{L}_{HFC}(v,v^{+},v^{-}) &= -2\alpha\cdot \mathbb{E}_{v,v^{+}}[f_{\theta}(v)^{T}f_{\theta}(v^{+})]\\
    &\ +\mathbb{E}_{v,v^{-}}[(f_{\theta}(v)^{T}f_{\theta}(v^{-}))^{2}],
    \end{aligned}
\end{equation} 
where $\alpha$ is used to control the rate of HFC within the graph.

Upon initial examination, one might observe that our $\mathcal{L}_{HFC}$ formulation closely aligns with $\mathcal{L}_{Spectral}$. Remarkably, the primary distinction lies in the introduction of the parameter $\alpha$. However, this is not a mere trivial addition; it emerges from intricate mathematical deliberation and is surprisingly consistent with $\mathcal{L}_{Spectral}$ that offers a nuanced control of the HFC rate within the graph. Minimizing our $\mathcal{L}_{HFC}$ results in more distinctive embeddings that strike a balance between LFC and HFC. Please kindly refer to Appendix~\ref{sec:spectral} for detailed discussions and proof.

\section{Methodology}
\label{sec:method}
\subsection{Overview}

\begin{figure*}
    \centering
    \includegraphics[width=0.9\textwidth]{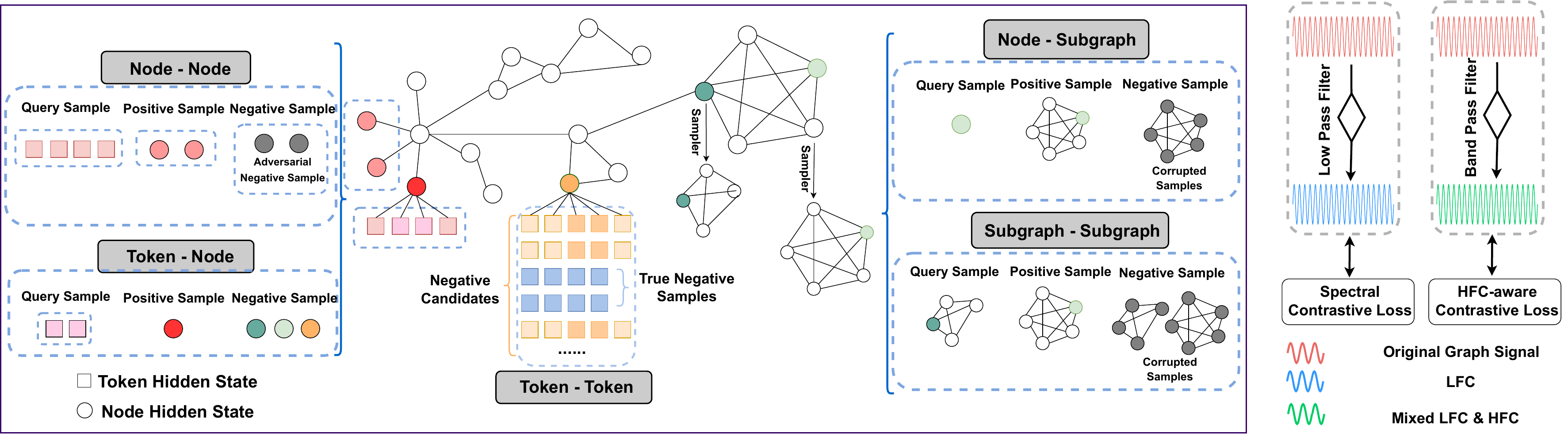}
    \caption{The overall architecture of HASH-CODE. 
    % We take the unidirectional-simplified GraphFormers trained with the two-stage progressive learning as our base model. 
    With GraphFormers as our base model, we incorporate five self-supervised learning objectives based on the HFC-aware contrastive loss to capture the text-graph correlations in different granularities. Spectral contrastive loss learns the LFC while our HFC-aware loss achieves the balance between HFC and LFC.}
    \label{fig:main}
\end{figure*}

Existing studies~\cite{jin2021bite,li2021adsgnn,zhang2019shne,zhu2021textgnn} mainly emphasize the effect of sequential and graphic characteristics using the supervised optimization objective alone.
Inspired by recent progress with contrastive learning~\cite{chen2020simple,he2020momentum}, we take a different perspective to characterize the data correlations by contrasting different views of the raw data.

The basic idea of our approach is to incorporate several elaborately designed self-supervised learning objectives for enhancing the original GNN and PLM. To develop such objectives, we leverage effective correlation signals reflected in the intrinsic characteristics of the input. As shown in Figure~\ref{fig:main}, for our task, we consider the information in different
levels of granularity, including token, node and sub-graph, which are considered as different views of
the input. By capturing the multi-view correlation, we unify these self-supervised learning objectives with the typical joint learning training scheme in language modeling and graph mining~\citep{yang2021graphformers}.

% The overview of our proposed method is presented in . We take the unidirectional-simplified GraphFormers~\citep{yang2021graphformers} trained with the two-stage progressive learning as our base model. In the following sections, we will
% describe how we utilize the correlation signals among tokens, nodes and sub-graphs to enhance the data representations based on our proposed HFC method. Finally, we present the discussions on our approach.

\subsection{Hierarchical Contrastive Learning with TAGs}
TAGs naturally possess 3 levels in the hierarchy: token-level, node-level and subgraph-level.
Based on the above GNN and PLM model, we further incorporate additional self-supervised signals with contrastive learning to enhance the representations of input data. We adopt a joint-training way to construct
different loss functions based on the multi-view correlation.

\subsubsection{Intra-hierarchy contrastive learning}\

\noindent \textbf{Modeling Token-level Correlations.}  We first begin with modeling the bidirectional information in the token sequence. Inspired by the masked language model like BERT~\citep{devlin2018bert}, we propose to use the contrastive learning framework to design a task that maximizes the mutual information between the masked sequence representation and its contextual representation vector. Specifically, for a node $v$, given its textual attribute sequence $x_{v} = \{x_{v,1}, x_{v,2},..., x_{v,T} \}$, we consider $x_{v,i:j}$ and $\hat{x}_{v,i:j}$ as a positive pair, where $x_{v,i:j}$ is an \textit{n}-grams spanning from i to j and $\hat{x}_{v,i:j}$ is the corresponding sequence masked at position i to j. We may omit the subscript $v$ for notation simplification when it is not important to differentiate the affiliation between node and textual sequence.

% The gist of token-token contrastive learning is to embed similar tokens nearby in the latent space while embedding those dissimilar ones far apart. However, the definition of dissimilar (\textit{i.e.,} negative) token pairs is non-trivial. Previous methods usually derive negative samples by sampling uniformly over the dataset~\cite{kong2019mutual}. However, they cannot guarantee that the produced negative samples own exactly distinct semantics relative to the query sample. Such a defect hampers token-token contrastive learning, in which those semantically relevant positive candidates could be wrongly expelled from the query sample in the latent space, and the semantic structure is thus broken to some extent. To overcome this drawback, we aim to select more precise negative samples that own truly irrelevant semantics.

For a specific query \textit{n}-gram $x_{i:j}$, instead of contrasting it indiscriminately with all negative candidates $\mathcal{N}$ in a batch~\cite{zhao2023beyond,wang2022adaptive,li2019adversarial,kong2019mutual}, we select truly negative samples for contrasting based on the supervision signals provided by the hierarchical structure in TAGs, as shown in Figure~\ref{fig:tt}. Intuitively, we would like to eliminate those candidates sharing highly similar semantics with the query, while keeping the ones that are less semantically relevant to the query. To achieve this goal, we first define a similarity measure between an \textit{n}-gram and a node. Inspired by~\cite{li2020prototypical}, for a node $v$, we define the semantic similarity between \textit{n}-gram's hidden state $h_{x_{i:j}}$ and this node's hidden state $h_{v}$ using a node-specific dot product:
\begin{equation}
    s(h_{x_{i:j}},h_{v}) = \frac{h_{x_{i:j}}\cdot h_{v}}{\tau_{h_{v}}},
    \tau_{h_{v}} = \frac{\Sigma_{h_{x_{i}}\in H_{v}}||h_{x_{i}}-h_{v}||_{2}}{|H_{v}|log(|H_{v}|+\epsilon)}\nonumber,
\end{equation}
where $h_{x_{i}}$ is the hidden representation of the token $x_{i}$, $H_{v}$ consists of the hidden representations of the tokens assigned to node $v$, and $\epsilon$ is a smooth parameter balancing the scale of temperature $\tau_{h_{v}}$ among different nodes.

% *****************************************
\begin{figure}[t]
    \centering
    \includegraphics[width=.8\linewidth]{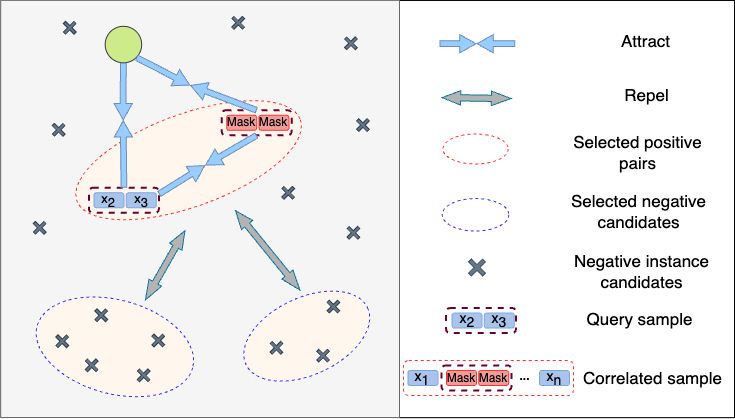}
    \caption{Token-level contrastive selective coding.}
    \label{fig:tt}
    % \vspace{-0.3cm}
\end{figure}
% ******************************************

On such a basis, we conduct negative sampling selection considering both the token and node hierarchies. Given the query \textit{n}-gram $x_{i:j}$, we denote its corresponding node $v$'s representation as $h_{v}$. For a negative candidate, we are more likely to select it if its similarity with $h_{v}$ is less prominent compared with other negative candidates' similarities with $h_{v}$. Based on such an intuition, the least dissimilar negative samples $\mathcal{N}_{select}(h_{x_{i:j}})$ are produced for the specific query.

By using these refined negative samples, we define the objective function of token-level contrastive (TC) loss as below:
\begin{equation}
    % \mathcal{L}_{TC}=E_{p(x_{i:j},\hat{x}_{i:j})}[L_{contrast}(g_{\omega}(x_{i:j}),g_{\omega}(\hat{x}_{i:j}),N_{select}(z),\tau)],
    \mathcal{L}_{TC}=\frac{1}{M}\Sigma_{m=1}^{M}\mathcal{L}_{HFC}(x_{m,i:j},\hat{x}_{m,i:j},\mathcal{N}_{select}(h_{x_{m,i:j}})),
\end{equation}
where $M$ is the size of the representation set and $\mathcal{L}_{HFC}$ is our proposed HFC-aware spectral contrastive loss.
% where $g_{\omega}$ is our base PLM.

% Therefore, we replace the spectral contrastive loss with our proposed HFC-aware contrastive loss to incorporate more HFC for learning discriminative representations. The objective function of token-level contrastive (TC) loss becomes:
% \begin{equation}
%     % \mathcal{L}_{TC}=E_{p(x_{i:j},\hat{x}_{i:j})}[L_{contrast}(g_{\omega}(x_{i:j}),g_{\omega}(\hat{x}_{i:j}),N_{select}(z),\tau)],
%     \mathcal{L}_{TC}=\frac{1}{M}\Sigma_{m=1}^{M}\mathcal{L}_{HFC}(x_{m,i:j},\hat{x}_{m,i:j},N_{select}(z)),
% \end{equation}
% where $M$ is the size of the representation set and $\mathcal{L}_{HFC}$ is the HFC-aware contrastive loss. Note that in the follow-up contrastive objective design, we use the proposed HFC-aware contrastive loss as the underlying contrastive loss.

% ****************************************
\begin{figure}[h]
    \centering
    \includegraphics[width=.8\linewidth]{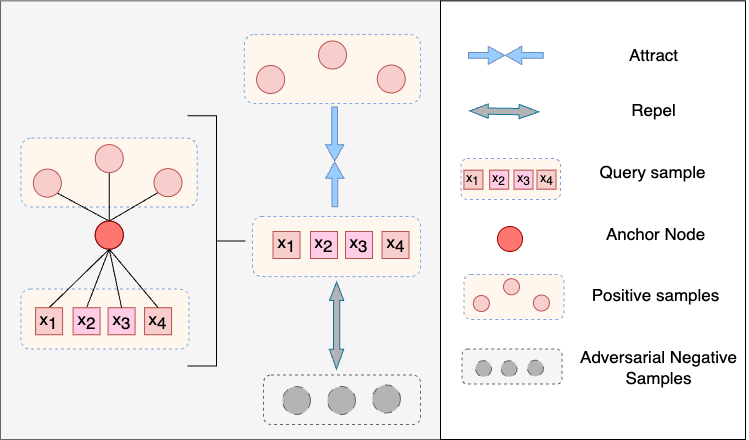}
    \caption{Modeling node-level correlations.}
    \label{fig:nn}
    % \vspace{-0.3cm}
\end{figure}
% *****************************************

\noindent \textbf{Modeling Node-level Correlations.}
Investigating the cross-view contrastive mechanism is especially important for node representation learning~\citep{wang2021self}. As mentioned before, nodes in TAGs possess textual attributes that can indicate semantic relationships in the network and serve as complementary to structural patterns. As shown in Figure~\ref{fig:nn}, given a node $v$, we treat its textual attribute sequence $x_{v}$ and its direct connected neighbors $u, \textnormal{for}~u\in N_{v}$ as two different views.

The negative selective encoding strategy in token-level correlation modeling tends to select less challenging negative samples, reducing their contribution over time. Inspired by~\citep{xia2022progcl}, we adopt the ProGCL~\citep{xia2022progcl} method to adversarially reweight and generate harder negative samples $\widetilde{v}$ using the mixup operation~\citep{zhang2017mixup}. Therefore, we minimize the following Node-level Contrastive (NC) loss:

% The negative selective encoding strategy used in token-level correlation modeling may select those easy negative samples that contribute less and less during the training process. 
% Inspired by~\citep{xia2022progcl}, we propose to adversarially generate the negative samples $\widetilde{v}$ in the node-level contrastive learning process. Specifically, we adopt ProGCL~\citep{xia2022progcl} method to reweight the negative node samples and performing mixup operation~\citep{zhang2017mixup} to generate hard negative samples $\widetilde{v}$. Therefore, we minimize the following Node-level Contrastive (NC) loss:
\begin{equation}
    \mathcal{L}_{NC}=\frac{1}{M}\Sigma_{m=1}^{M}\mathcal{L}_{HFC}(x_{m,v},N_{m,v},\widetilde{v_{m}})
\end{equation}

\noindent \textbf{Modeling Subgraph-level Correlations.}
Having analyzed the correlations between a node’s local neighborhood and its textual attributes, we extend our investigation to encompass the correlations among subgraphs. This approach facilitates the representation of both local and higher-order structures associated with the nodes. It is reasonable that nodes are more strongly correlated with their immediate neighborhoods than with distant nodes, which exert minimal influence on them. Consequently, local communities are likely to emerge within the graph. Therefore, to facilitate our analysis, we select a series of subgraphs that include regional neighborhoods from the original graph to serve as our training data.

% Having modeled correlations between a node's local neighborhood and its textual features, we further consider modeling the correlations between subgraphs to cover both of the local and high-order structures of the nodes. Intuitively, nodes and their regional neighborhoods
% are more correlated while long-distance nodes hardly influence them. Therefore, local communities may form with the graph. This assumption is more reasonable as the size of graphs increases. Therefore, we sample a series of subgraphs including regional neighborhoods from the original graph as training data.

The paramount challenge currently lies in sampling a context subgraph that can furnish adequate structural information essential for the derivation of a high-quality representation of the central node. In this context, we adopt the subgraph sampling methodology based on the personalized PageRank algorithm (PPR)~\citep{jeh2003scaling} as introduced in~\citep{zhang2020graph,jiao2020sub}. Given the variability in the significance of different neighbors, for a specific node $i$, the subgraph sampler $S$ initially computes the importance scores of neighboring nodes utilizing PPR. Considering the relational data among all nodes represented by an adjacency matrix $A\in \mathbb{R}^{N\times N}$, the resulting matrix $S$ of importance scores is designated as 

% The most critical issue now is to sample a context subgraph, which can provide sufficient structure information for learning a high-quality representation for the central node. Here we follow the subgraph sampling based on personalized PageRank  algorithm (PPR)~\citep{jeh2003scaling} as introduced in~\citep{zhang2020graph,jiao2020sub}. Considering the
% importance of different neighbors varies, for a specific node $i$, the subgraph sampler $S$ first measures the importance scores of other neighbor nodes by PPR. Given the relational information between all nodes in the form of an adjacency matrix, $A\in \mathbb{R}^{N\times N}$, the importance score
% matrix $S$ can be denoted as
\begin{equation}
    S=\alpha\cdot(I-(1-\alpha)\cdot\overline{A})\nonumber,
\end{equation}
where $I$ represents the identity matrix and $\alpha$ is a parameter within the range $[0,1]$. The term $\overline{A}=AD^{-1}$ is the column-normalized adjacency matrix, where $D$ is the corresponding diagonal matrix with entries $D(i,i)=\Sigma_{j}A(i,j)$ along its diagonal. The vector $S(i,:)$ enumerates the importance scores for node $i$.

% where $I$ is the identity matrix and $\alpha\in [0,1]$ is a parameter that is always set as 0.15. $\overline{A}=AD^{-1}$ denotes the colum-normalized adjacency matrix, where $D$ denotes the corresponding diagonal matrix with $D(i,i)=\Sigma_{j}A(i,j)$ on its diagonal. $S(i,:)$ is the importance scores vector for node $i$, indicating its correlation with other nodes.

% It is noted that the importance score matrix S can be
% precomputed before model training starts. And we implement node-wise PPR to calculate importance scores to reduce computation memory, which makes our method more suitable to work on large-scale graphs.

For a specific node $i$, the subgraph sampler $S$ selects the top-k most significant neighbors to form the subgraph $G_{i}$. The indices of the selected nodes are
% For a specific node $i$, the subgraph sampler $S$ chooses top-k important neighbors to constitute the subgraph $G_{i}$. The index of chosen nodes can be denoted as
\begin{equation}
    idx = top\_rank(S(i,:), k)\nonumber,
\end{equation}
where $top\_rank$ is the function that returns the indices corresponding to the top-k values, where k specifies the size of the context graphs.

The subgraph sampler $S$ processes the original graph along with the node index to derive the context subgraph $G_{i}$ for node $i$. The adjacency matrix $A_{i}$ and feature matrix $X_{i}$ of this subgraph are defined as follows:
\begin{equation}
    A_{i}=A_{idx,idx,} X_{i}=X_{idx,:,}\nonumber
\end{equation}
where $.idx$ denotes an indexing operation. $A_{idx,idx}$ refers to the adjacency matrix, row-wise and column-wise indexed to correspond to the induced subgraph.  $X_{idx,:}$ is the feature matrix indexed row-wise. 

\noindent \textbf{Encoding subgraph.}
Upon acquiring the context subgraph $G_{i}=(A_{i},X_{i})$ of a central node $i$, the encoder $\mathcal{E}:\mathbb{R}^{N\times N}\times \mathbb{R}^{N\times F}\rightarrow \mathbb{R}^{N\times F}$ encodes it to derive the latent representations matrix $H_{i}$, which is denoted as
\begin{equation}
    H_{i} = \mathcal{E}(A_{i},X_{i})\nonumber
\end{equation}

The subgraph-level representation $s_{i}$ is summarized using a readout function, $\mathcal{R}:\mathbb{R}^{N\times F}\rightarrow\mathbb{R}^{F}$:
\begin{equation}
    s_{i}=\mathcal{R}(H_{i})\nonumber
\end{equation}.

% ******************************************
\begin{figure}[t]
    \centering
    \includegraphics[width=.8\linewidth]{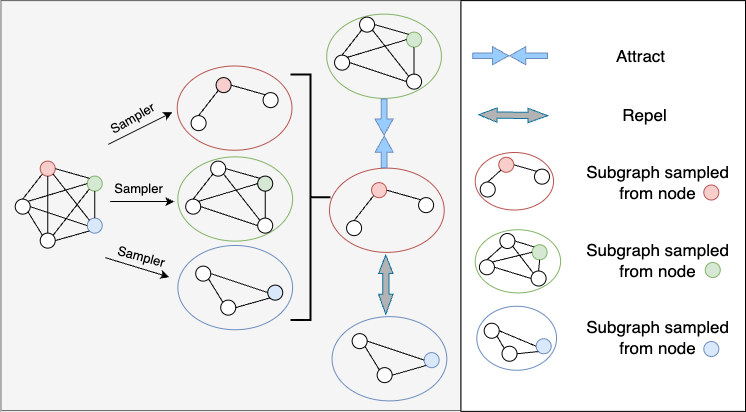}
    \caption{Modeling subgraph-level correlations.}
    \label{fig:ss}
    % \vspace{-0.3cm}
\end{figure}
% ******************************************

So far, the representations of subgraphs have been produced. As shown in Figure~\ref{fig:ss}, to model the correlations in subgraph level, we treat two subgraphs $s_{i}$ and $\hat{s}_{i}$ that sampled from the node $h_{i}$ and its most important neighbor node $\hat{h}_{i}$ respectively as positive pairs while the rest of subgraphs $\widetilde{s}$ are negative pairs. We minimize the following Subgraph-level Contrastive (SC) loss:
\begin{equation}
    \mathcal{L}_{SC} = \frac{1}{M}\Sigma_{m=1}^{M}\mathcal{L}_{HFC}(s_{m}, \hat{s}_{m},\widetilde{s_{m}})
\end{equation}

% ***********************************
% \begin{figure}[t]
%     \centering
%     \includegraphics[width=.8\linewidth]{figures/Token-Node.png}
%     \caption{Modeling token-node correlations.}
%     \label{fig:tn}
%     % \vspace{-0.3cm}
% \end{figure}
% **************************************
\subsubsection{Inter-hierarchy contrastive learning}\

\noindent Having modeled the intra-hierarchy correlations, we further consider modeling the intra-hierarchy correlations as different hierarchies are dependent and will influence each other.

\noindent \textbf{Modeling Token-Node Correlations.}
To model the token-node correlation, our intuition is to train the language model to refine the understanding of the text by GNN produced embeddings. Therefore, the language model is pushed to learn fine-grained task-aware context information. Specifically, 
% as shown in Figure~\ref{fig:tn}, 
given a sequence $x_{v} = \{x_{v,1}, x_{v,2},..., x_{v,T} \}$, we consider $x_{v}$ and its corresponding node representation $h_{v}$ as a positive pair. On the other hand, for a set of node representations, we employ a function, $\mathcal{P}$, to corrupt them to generate negative samples, denoted as
\begin{equation}
    \{\widetilde{h_{1}}, \widetilde{h_{2}},...,\widetilde{h_{M}}\} = \mathcal{P}\{h_{1}, h_{2},...,h_{M}\}\nonumber,
\end{equation}
where $M$ is the size of the representation set. $\mathcal{P}$ is the random shuffle function in our experiment. This corruption strategy determines the differentiation of tokens with different context nodes, which is crucial for some downstream tasks, such as node classification.
We develop the following Token-Node Contrastive (TNC) loss:
\begin{equation}
    % \mathcal{L}_{TNC} = \frac{1}{M}\Sigma_{k=1}^{M}\mathbb{E}_{}[\mathcal{L}_{HFC}(x_{i:j},h_{k},\widetilde{h_{k}})]
    \mathcal{L}_{TNC} = \frac{1}{M}\Sigma_{m=1}^{M}\mathcal{L}_{HFC}(x_{m,v},h_{m,v},\mathcal{P}\{h_{1}, h_{2},...,h_{M}\})
\end{equation}

% % ****************************************
% \begin{figure}[h]
%     \centering
%     \includegraphics[width=.8\linewidth]{figures/Node-Sub.png}
%     \caption{Modeling node-subgraph correlations.}
%     \label{fig:ns}
%     \vspace{-0.3cm}
% \end{figure}
% % ****************************************

\noindent \textbf{Modeling Node-Subgraph Correlations.} Intuitively, nodes are dependent on their regional neighborhoods and different nodes have different context subgraphs. Therefore, we consider the strong correlation between central nodes and their context subgraphs to design a self-supervision pretext task to contrast the real context subgraph with a fake one. Specifically, 
% as shown in Figure~\ref{fig:ns}, 
for the node representation, $h_{v}$, that captures the regional information in the context subgraph, we regard the context subgraph representation $s_{v}$ as the positive sample. Similar to the calculation of $\mathcal{L}_{TNC}$, we employ the random shuffle function $\mathcal{P}$ to corrupt other subgraph representations to generate negative samples, denoted as
\begin{equation}
    \{\widetilde{s_{1}}, \widetilde{s_{2}},...,\widetilde{s_{M}}\} = \mathcal{P}\{s_{1}, s_{2},...,s_{M}\}\nonumber
\end{equation}
We minimize the following Node-Subgraph Contrastive (NSC) loss:
\begin{equation}
    \mathcal{L}_{NSC} =\frac{1}{M}\Sigma_{m=1}^{M} \mathcal{L}_{HFC}(h_{m,v},s_{m,v},\mathcal{P}\{s_{1}, s_{2},...,s_{M}\})
\end{equation}

\noindent \textbf{Overall Objective Loss.} Our overall objective function is a weighted combination of the five terms above:
\begin{equation}
\begin{aligned}
    \mathcal{L}_{HASH-CODE}&=\lambda_{TC}\mathcal{L}_{TC}+\lambda_{NC}\mathcal{L}_{NC}+\lambda_{SC}\mathcal{L}_{SC}\\
    &+\lambda_{TNC}\mathcal{L}_{TNC}+\lambda_{NSC}\mathcal{L}_{NSC},
\end{aligned}
\end{equation}
where $\lambda_{TC},\lambda_{NC},\lambda_{SC},\lambda_{TNC}$ and $\lambda_{NSC}$ are hyper-parameters that balance the contribution of each term. 
% We summarize the workflow of our proposed HASH-CODE in Appendix~\ref{sec:workflow}.

\section{Experiments}
\label{sec:exp}
\subsection{Experimental Setup}
% In this section, we have conducted extensive experiments, and analyzed the performance of the proposed HASH-CODE method.
% by addressing the following key research questions as follows:
% \begin{itemize}[leftmargin=*,noitemsep,topsep=0pt]
%     \item \textbf{RQ1:} How does our method perform compared with baseline methods?
%     \item \textbf{RQ2:} How does each component of our method contribute to the performance?
%     \item \textbf{RQ3:} How about the efficiency of our proposed model compared with other baselines?
%     \item \textbf{RQ4:} How does our method perform when facing the issue of data sparsity?
%     \item \textbf{RQ5:} How do different hyper-parameters affect our method?
% \end{itemize}

\subsubsection{Datasets}
We conduct experiments on six datasets (\textit{i.e.,} DBLP\footnote{\noindent  \url{https://originalstatic.aminer.cn/misc/dblp.v12.7z}}, Wikidata5M\footnote{\noindent   \url{https://deepgraphlearning.github.io/project/wikidata5m}}~\citep{wang2021kepler}, Beauty, Sports and Toys from Amazon dataset\footnote{\noindent  \url{ http://snap.stanford.edu/data/amazon/}}~\citep{mcauley2015image} and Product Graph) from three different domains (i.e., academic papers, social media posts, and e-commerce). We
leverage three common metrics to measure the prediction accuracy: Precision@1 (P@1), NDCG, and MRR. 
% Detailed information about the datasets can be found in Appendix~\ref{sec:data}. 
The statistics of the six datasets are summarized in Table~\ref{tab:dataset}.

\begin{table}[t]
    \centering
    \caption{Statistics of datasets after preprocessing.}
    \resizebox{0.5\textwidth}{!}{
    \begin{tabular}{ccccccc}
        \toprule
           Dataset      & Product   & Beauty & Sports & Toys & DBLP & Wiki\\
        \midrule
        \#Users  & 13,647,591 & 22,363 & 25,598 & 19,412 & N/A & N/A \\
        \#Items & 5,643,688 & 12,101 & 18,357 & 11,924 & 4,894,081 & 4,818,679  \\
        \#N   & 4.71 & 8.91 & 8.28 & 8.60 & 9.31 & 8.86  \\
        \#Train  & 22,146,934 & 188,451 & 281,332 & 159,111 & 3,009,506 & 7,145,834  \\
        \#Valid  & 30,000 & 3,770 & 5,627 & 3,182&  60,000 & 66,167  \\
        \#Test   & 306,742 & 6,280 & 9,377 & 5,304 & 100,000  & 100,000  \\
        \bottomrule
    \end{tabular}}
    \label{tab:dataset}
    % \vspace{-0.4cm}
\end{table}

\subsubsection{Baselines}
We compare HASH-CODE with three types of baselines: (1) GNN-cascaded transformers, which includes BERT+Ma-xSAGE~\citep{hamilton2017inductive}, BERT+MeanSAGE~\citep{hamilton2017inductive},  BERT+GAT~\citep{velivckovic2017graph}, TextGNN~\citep{zhu2021textgnn}, and AdsGNN~\citep{li2021adsgnn}. (2) GNN-nested transformers, which includes GraphFormers~\citep{yang2021graphformers}, and Heterformer~\citep{jin2022heterformer}. (3)
Vanilla GraphSAGE~\citep{hamilton2017inductive}, Vanilla GAT~\citep{velivckovic2017graph}, Vanilla BERT~\citep{devlin2018bert} and Twin-Bert~\citep{lu2020twinbert}. 
% Detailed information about the baselines can be found in Appendix~\ref{sec:baseline}.

\subsubsection{Reproducibility.} 
For all compared models, we adopt the 12-layer BERT-base-uncased~\citep{devlin2018bert} in the
huggingface as the backbone PLM for a fair comparison. The models are
trained for at most 100 epochs on all datasets. We use an early
stopping strategy on P@1 with a patience of 2 epochs. The size of minimal training batch is
64, learning rate is set to $1e-5$. We pad the sequence length to 32 for Product, DBLP and Amazon datasets, 64 for Wiki, depending on different text length of each dataset. Adam optimizer~\citep{kingma2014adam} is employed to minimize the training loss. Other parameters are tuned on the validation dataset and we save the checkpoint with the best validation performance as the final model. Parameters in baselines are carefully tuned on the validation set to select the most desirable parameter setting.

\begin{table*}[h]
\huge
\renewcommand{\arraystretch}{1.5}
\centering
  \caption{Experiment results of link prediction. The results of the best performing baseline are underlined.
  % (HASH-CODE marked in bold, the best baseline underlined). 
  The numbers in bold indicate statistically significant improvement (p < .01) by the pairwise t-test comparisons over the other baselines.}
  \label{tab:main}
  \resizebox{1.0\textwidth}{!}{
  \begin{tabular}{ccccccccccccccc}
    \toprule
    Datasets & Metric & MeanSAGE & GAT & Bert & Twin-Bert & Bert+MeanSAGE & Bert+MaxSAGE & Bert+GAT & TextGNN &  AdsGNN & GraphFormers & Heterformer & HASH-CODE & Improv. \\ 
    \midrule
    \multirow{3}{*}{Product} & P@1 & 0.6071 & 0.6049 & 0.6563 & 0.6492 & 0.7240 & 0.7250 & 0.7270 & 0.7431 & 0.7623 & 0.7786 & \underline{0.7820} & $\textbf{0.7967}^{*}$ & 1.88\%  \\ & NDCG & 0.7384 & 0.7401 & 0.7911 & 0.7907 & 0.8337 & 0.8371 & 0.8378 & 0.8494 & 0.8605 & 0.8793 & \underline{0.8861} & $\textbf{0.9039}^{*}$ & 2.01\% \\ & MRR & 0.6619 & 0.6627 & 0.7344 & 0.7285 & 0.7871 & 0.7832 & 0.7880 & 0.8107 & 0.8361 & 0.8430 & \underline{0.8492} & $\textbf{0.8706}^{*}$ & 2.52\% \\
    \midrule
    \multirow{3}{*}{Beauty} & P@1 & 0.1376 & 0.1367 & 0.1528 & 0.1492 & 0.1593 & 0.1586 & 0.1544 & 0.1625 & 0.1669 & \underline{0.1774} & 0.1739 & $\textbf{0.1862}^{*}$ & 4.96\% \\ & NDCG & 0.2417 & 0.2469 & 0.2702 & 0.2683 & 0.2741 & 0.2756 & 0.2726 & 0.2863 & 0.2891 & \underline{0.2919} & 0.2911 & $\textbf{0.3061}^{*}$ & 4.86\% \\ & MRR & 0.2558 & 0.2549 & 0.2680 & 0.2638 & 0.2712 & 0.2759 & 0.2720 & 0.2802 & 0.2821 & \underline{0.2893} & 0.2841 & $\textbf{0.3057}^{*}$ & 5.67\% \\
    \midrule
    \multirow{3}{*}{Sports} & P@1 & 0.1102 & 0.1088  & 0.1275 & 0.1237 & 0.1330 & 0.1311 & 0.1302 & 0.1421 & 0.1466 & \underline{0.1548} & 0.1534 & $\textbf{0.1623}^{*}$ & 4.84\% \\ & NDCG & 0.2091 & 0.2116 & 0.2375 & 0.2297 & 0.2432 & 0.2478 & 0.2419 & 0.2537 & 0.2582 & 0.2674 & \underline{0.2692} & $\textbf{0.2775}^{*}$ & 3.08\% \\ & MRR & 0.2171 & 0.2168 & 0.2319 & 0.2296 & 0.2434 & 0.2471 & 0.2397 & 0.2612 & 0.2653 & \underline{0.2679} & 0.2640 & $\textbf{0.2754}^{*}$ & 2.80\% \\
    \midrule
    \multirow{3}{*}{Toys} & P@1 & 0.1342 & 0.1330 & 0.1498 & 0.1427 & 0.1520 & 0.1536 & 0.1514 & 0.1658 & 0.1674 & \underline{0.1703} & 0.1685 & $\textbf{0.1767}^{*}$ & 3.76\% \\ & NDCG & 0.2015 & 0.2028 & 0.2249 & 0.2206 & 0.2451 & 0.2486 & 0.2413 & 0.2692 & 0.2734 & \underline{0.2859} & 0.2823 & $\textbf{0.2946}^{*}$ & 3.04\% \\ & MRR & 0.2173 & 0.2149 & 0.2311 & 0.2276 & 0.2509 & 0.2527 & 0.2476 & 0.2648 & 0.2715 & \underline{0.2803} & 0.2778 & $\textbf{0.2919}^{*}$ & 4.14\% \\
    \midrule
    %细数数据集上的improv.更大
    \multirow{3}{*}{DBLP} & P@1 & 0.4963 & 0.4931 & 0.5673 & 0.5590 & 0.6533 & 0.6596 & 0.6634 & 0.6913 & 0.7102 & 0.7267 & \underline{0.7288} & $\textbf{0.7446}^{*}$ & 2.17\% \\ & NDCG & 0.6997 & 0.6981 & 0.7484 & 0.7417 & 0.8004 & 0.8059 & 0.8086 & 0.8331 & 0.8507 & 0.8565 & \underline{0.8576} & $\textbf{0.8823}^{*}$ & 2.88\% \\ & MRR & 0.6314 & 0.6309 & 0.6777 & 0.6643 & 0.7266 & 0.7067 & 0.7300 & 0.7792 & 0.7805 & 0.8133 & \underline{0.8148} & $\textbf{0.8428}^{*}$ & 3.44\% \\
    \midrule
    \multirow{3}{*}{Wiki} & P@1 & 0.2850 & 0.2862 & 0.3066 & 0.3015 & 0.3306 & 0.3264 & 0.3412 & 0.3693 & 0.3820 & \underline{0.3952} & 0.3947 & $\textbf{0.4104}^{*}$ & 3.85\% \\ & NDCG & 0.5389 & 0.5357 & 0.5699 & 0.5613 & 0.5730 & 0.5737 & 0.6071 & 0.6098 & 0.6155 & 0.6230 & \underline{0.6233} & $\textbf{0.6402}^{*}$ & 2.71\% \\ & MRR & 0.4411 & 0.4436 & 0.4712 & 0.4602 & 0.4980 & 0.4970 & 0.5022 & 0.5097 & 0.5134 & \underline{0.5220} & 0.5216 & $\textbf{0.5356}^{*}$ & 2.61\% \\ \bottomrule
  \end{tabular} }
  % \vspace{-0.3cm}
\end{table*}

\subsection{Overall Comparison}
Following previous studies~\citep{yang2021graphformers,jin2022heterformer} on network representation learning,
we consider two fundamental task:
link prediction and node classification. 
To save space, we will mainly present the results on link prediction here and save the node classification part to Appendix~\ref{sec:node_classification}. The overall evaluation results are reported in Table~\ref{tab:main}. We have the following observations:

% \noindent \textbf{Settings.} 
% The link prediction experiments are evaluated in terms of link prediction accuracy, i.e., to predict whether a query node and key node are connected given the textual features of themselves and their neighbours. For Product, DBLP and Wiki datasets, in each testing instance, one query is provided with 300 keys: 1 positive plus 299 randomly sampled negative cases.

% We leverage three common metrics to measure the prediction accuracy: Precision@1, MRR, and NDCG. Given a query node $u$, Precision@1 measures whether the key node $v$ linked with $u$ is ranked the highest in the batch; MRR calculates the average of the reciprocal ranks of $v$; NDCG further takes the order and relative importance of $v$ into account and here we calculate on the full candidate list, the length of which equals to test batch size.

% \noindent \textbf{Results.} 

In comparing vanilla textual and graph models across various datasets, we find a consistent performance ranking: BERT outperforms Twin-BERT, which in turn exceeds GAT and GraphSAGE. This hierarchy reveals GNN models' limitations in capturing rich textual semantics due to their focus on node proximity and global structural information. Specifically, the superior performance of the one-tower BERT model over the two-tower Twin-BERT model underscores the advantage of integrating information from both sides, despite BERT's potential inefficiency in low-latency scenarios due to one-by-one similarity computations.

% For four vanilla textual/graph baselines, the performance order is consistent across all datasets, \textit{i.e.,} $\textnormal{Bert} > \textnormal{Twin-Bert} > \textnormal{GAT} \approx \textnormal{GraphSAGE}$. GNN models obtain the worst performance, as they can only model the node proximity that preserved by the global structural information, but fail to encode the textual information that presents rich semantics to characterize the property of each node. This demonstrates the importance of leveraging the local textual information of individual nodes. As for the vanilla textual baselines, the one-tower textual model (BERT) outperforms the two-tower model (Twin-BERT) as it can incorporate the information from both sides, while two-tower models can only exploit the data from a single side. However, one-tower structure has to compute the similarity between a search query and each ad one-by-one, which is not suitable for low-latency online scenario. In general, vanilla textual/graph models  perform worse than GNN-cascaded transformers, which demonstrates the importance of encoding both text and network signals in text-attributed graphs.

As for GNN-cascaded transformers, BERT+GAT generally surpasses BERT+MeanSAGE and BERT+MaxSAGE in modeling attributes on the Product, DBLP, and Wiki datasets, attributed to its multi-head self-attention mechanism. However, its performance dips on the Beauty, Sports, and Toys datasets, likely due to noise from keyword-based attributes in Amazon Reviews. Despite these variations, GNN-cascaded transformers fall short of co-training methods, largely because of the static nature of node textual features during training. Among the models, AdsGNN consistently leads over TextGNN across all datasets. This highlights the effectiveness of AdsGNN's node-level aggregation model in capturing the nuanced roles of queries and keys, proving a tightly-coupled structure's superiority in integrating graph and textual data.

% As for GNN-cascaded transformers, Bert+GAT performs better than Bert+MeanSAGE and Bert+MaxSAGE on Product, DBLP and Wiki datasets, because the multi-head self-attention mechanism has a stronger capacity to model attributes. However, the performance of GAT is worse than that of MeanSAGE on Beauty, Sports and Toys datasets. A potential reason is that the multi-head self-attention may incorporate more noise from the attributes since they are keywords extracted from the reviews on Amazon Reviews. In general, GNN-cascaded transformers perform worse than co-training-based methods, which may be due to the node textual features are pre-existed and fixed in the training phase, leading to the limited expression capacity. AdsGNN consistently outperforms TextGNN on all datasets. This is because compared with TextGNN, the node-level aggregation model AdsGNN can capture the different roles of queries and keys, demonstrating that the tightly-coupled structure is more powerful than the loosely-coupled framework in deeply fusing the graph and textual information.

For GNN-nested transformers, Heterformer outperforms Graphformers on denser networks like those of the Product and DBLP datasets compared to the Amazon datasets. Our HASH-CODE consistently outshines all baselines, achieving 2\%$\sim$4\% relative improvements on six datasets against the most competitive ones.  These findings affirm the efficacy of contrastive learning in enhancing co-training architectures for representation learning tasks.

\begin{figure}[h]
\subfigure[P@1]{\centering
    \includegraphics[width=0.48\linewidth]{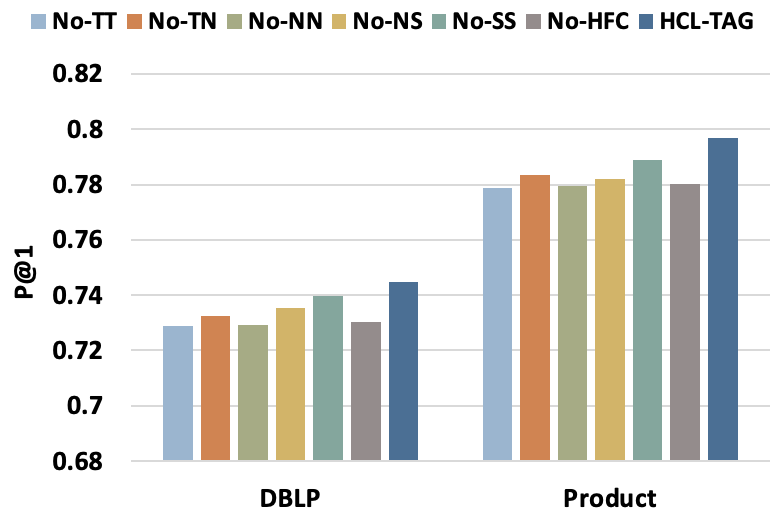}
    }
\subfigure[NDCG]{\centering
    \includegraphics[width=0.48\linewidth]{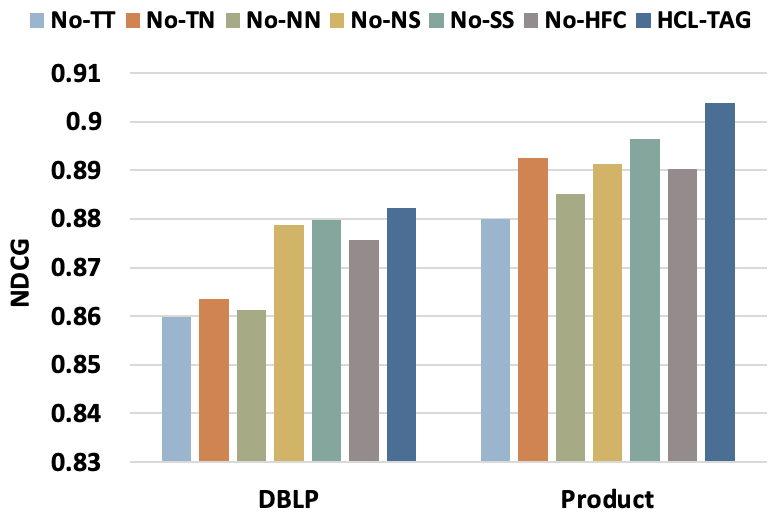}
    }
    \caption{Ablation studies of different components on DBLP and Products datasets.}
    \label{fig:ablation}
    % \vspace{-0.3cm}
\end{figure}
% \vspace{-0.2cm}

\begin{table*}[t]
% \Large
  \caption{Time and memory costs per mini-batch for GraphFormers and HASH-CODE, with neighbour size increased from 3 to 200. HASH-CODE achieve similar efficiency and scalability as GraphFormers.}
  \label{tab:efficiency}
  {
  \begin{tabular}{cccccccc}
    \toprule
    \#N & 3 & 5 & 10 & 20 & 50 & 100 & 200 \\
    \midrule
    Time: GraphFormers  & $63.95$ms  & $97.19$ms & $170.16$ms & $306.12$ms & $714.32$ms & $1411.09$ms & $2801.67$ms   \\
    Time: HASH-CODE & $67.68$ms & $105.35$ms & $180.03$ms &  $324.11$ms& $754.97$ms & $1573.29$ms & $2962.86$ms  \\
    \midrule
    Mem: GraphFormers & $1.33$GiB & $1.39$GiB & $1.55$GiB & $1.83$GiB & $2.70$GiB & $4.28$GiB & $7.33$GiB   \\
    Mem: HASH-CODE & $1.33$GiB & $1.39$GiB & $1.55$GiB & $1.84$GiB & $2.72$GiB & $4.43$GiB & $7.72$GiB   \\
    \bottomrule
  \end{tabular}
  }
\end{table*}

\subsection{Ablation Study}
Our proposed HASH-CODE designs five pre-training objectives based on the HFC-aware contrastive objective. In this section, we conduct the ablation study
on Product and DBLP datasets to analyze the contribution of each objective. 
We evaluate the performance of several HASH-CODE variants: (a) No-TT removes the $\mathcal{L}_{TC}$; (b) No-TN removes the $\mathcal{L}_{TNC}$; (c) No-NN removes the $\mathcal{L}_{NC}$; (d) No-NS removes the $\mathcal{L}_{NSC}$; (e) No-SS removes the $\mathcal{L}_{SC}$; (f) No-HFC replaces the HFC-aware loss with spectral contrastive loss. The results from GraphFormers are also provided for comparison. P@1 and NDCG@10 are adopted for this evaluation.

From Figure~\ref{fig:ablation}, we can observe that removing any contrastive learning objective would lead to the performance decrease, indicating all the objectives are useful to capture the correlations in varying levels of granularity in TAGs. Besides, the importance of these objectives is varying on different datasets. Overall, $\mathcal{L}_{TC}$ is more important than others. Removing it yields a larger drop of performance on all datasets, indicating that natural language understanding is more important on these datasets. In addition, No-HFC performs worse than the other variants, indicating the importance of learning more discriminative embeddings with HFC.

\subsection{Efficiency Analysis}
We compare the time efficiency between HASH-CODE, and GNN-nested Transformers (GraphFormers). The evaluation is conducted utilizing an Nvidia 3090 GPU. We follow the same setting with~\citep{yang2021graphformers}, where each mini-batch contains 32 encoding instances; each instance contains one center and \#N neighbour nodes; the token length of each node is 16. We report the average time and memory (GPU RAM) costs per mini-batch in Table~\ref{tab:efficiency}. 

We find that the time and memory costs associated with these methods exhibit a linear escalation in tandem with the augmentation of neighboring elements. Meanwhile, the overall time and memory costs of HASH-CODE exhibit a remarkable proximity to GraphFormers, especially when the number of neighbor nodes is small. In light of the above observations, it is reasonable to deduce that HASH-CODE exhibits superior accuracy while concurrently maintaining comparable levels of efficiency and scalability when juxtaposed with GNN-nested transformers.

% Firstly, the time and memory costs of these methods grow linearly with the increment of neighbours. (There are overheads of time and memory costs. The time cost overhead may come from CPU processing; while the memory cost overhead is mainly due to the model parameters~\citep{rajbhandari2020zero}). We may approximately remove the overheads by deducting the time and memory costs where \#N=3). 

% Secondly, the overall time and memory costs of HASH-CODE are quite close to GraphFormers. When the number of neighbour nodes is small, the differences between both methods are almost ignorable. The differences become slightly larger when more neighbour nodes are included. However, the differences are still relatively small: merely around 5.8\% of the overall running costs when \#N is increased to 200.

% Based on the above observations, we may conclude that HASH-CODE are more accurate, meanwhile equally efficient and scalable as the GNN-nested tranformers.

% \begin{figure}[h]
% \subfigure[DBLP]{\centering
%     \includegraphics[width=0.48\linewidth]{figures/DBLP_Data_Amount.pdf}
%     }
% \subfigure[Product]{\centering
%     \includegraphics[width=0.48\linewidth]{figures/Product_Data_Amount.pdf}
%     }
%     \caption{Performance (P@1) comparison w.r.t. different sparsity levels on DBLP and Product datasets. The performance substantially drops when less training data is used, while  HASH-CODE is consistently better than baselines in all
% cases, especially in an extreme sparsity level (20\%).}
%     \label{fig:sparsity}
%     % \vspace{-0.3cm}
% \end{figure}

\subsection{In-depth Analysis}
We continue to investigate several properties of the models 
in the next couple sections. To save space, we will mainly present the results here and save the details to the appendix:
\begin{itemize}[leftmargin=*,noitemsep,topsep=0pt]
    \item In Appendix~\ref{sec:sparsity}, we
simulate the data sparsity scenarios by using different proportions of the full dataset. We find that HASH-CODE is consistently better than baselines in all cases, especially in an extreme sparsity level (20\%). This observation implies that HASH-CODE is able to make better use of the data with the contrastive learning method, which alleviates the influence of data sparsity problem for representation learning to some extent.
    \item In Appendix~\ref{sec:epochs}, we investigate the influence of the number of training epochs on our performance. The results show that our model benefits mostly from the first 20 training epochs. And after that, the performance improves slightly. Based on this observation, we can conclude that the correlations among different views on TAGs can be well-captured by our contrastive learning
approach through training within a small number of epochs. So that the enhanced data representations can improve the performance of the downstream tasks.
    \item In Appendix~\ref{sec:neighbor_size}, we analyze the impact of neighbourhood size with a fraction of neighbour nodes randomly sampled for each center node. We can observe that with the increasing number of neighbour nodes, both HASH-CODE and Graphformers achieve higher prediction accuracies. However, the marginal gain is varnishing, as the relative improvement becomes smaller when more neighbours are included. In all the testing cases, HASH-CODE maintains consistent advantages over GraphFormers, which demonstrates the effectiveness of our proposed method.
    % \item In Appendix~\ref{sec:visualization}, we visualize the input node embeddings for different target classes by t-SNE ~\citep{van2008visualizing} to intuitively study the impact of our HFC-loss. We find that our $\mathcal{L}_{HFC}$ helps the model learn more discriminative node embeddings compared with $\mathcal{L}_{Spectral}$. 
\end{itemize}

\section{Conclusion}
\label{sec:con}
In this paper, we introduce the problem of node representation learning on TAGs and propose HASH-CODE, a hierarchical contrastive learning architecture to address the problem. Different from previous “cascaded architectures”, HASH-CODE utilizes five self-supervised optimization objectives to facilitate thorough mutual enhancement between network and text signals in different granularities. We also propose a HFC-aware spectral contrastive loss to learn more discriminative node embeddings. Experimental results on various graph mining tasks, including link prediction and node classification demonstrate the superiority of HASH-CODE. Moreover, the proposed framework can serve as a building block with different task-specific inductive biases. It would be interesting to see its future applications on real-world TAGs such as recommendation, abuse detection and tweet-based network analysis.
% \clearpage

\bibliographystyle{ACM-Reference-Format}
\balance
\bibliography{citation}

%%% -*-BibTeX-*-
%%% Do NOT edit. File created by BibTeX with style
%%% ACM-Reference-Format-Journals [18-Jan-2012].

\begin{thebibliography}{77}

%%% ====================================================================
%%% NOTE TO THE USER: you can override these defaults by providing
%%% customized versions of any of these macros before the \bibliography
%%% command.  Each of them MUST provide its own final punctuation,
%%% except for \shownote{}, \showDOI{}, and \showURL{}.  The latter two
%%% do not use final punctuation, in order to avoid confusing it with
%%% the Web address.
%%%
%%% To suppress output of a particular field, define its macro to expand
%%% to an empty string, or better, \unskip, like this:
%%%
%%% \newcommand{\showDOI}[1]{\unskip}   % LaTeX syntax
%%%
%%% \def \showDOI #1{\unskip}           % plain TeX syntax
%%%
%%% ====================================================================

\ifx \showCODEN    \undefined \def \showCODEN     #1{\unskip}     \fi
\ifx \showDOI      \undefined \def \showDOI       #1{#1}\fi
\ifx \showISBNx    \undefined \def \showISBNx     #1{\unskip}     \fi
\ifx \showISBNxiii \undefined \def \showISBNxiii  #1{\unskip}     \fi
\ifx \showISSN     \undefined \def \showISSN      #1{\unskip}     \fi
\ifx \showLCCN     \undefined \def \showLCCN      #1{\unskip}     \fi
\ifx \shownote     \undefined \def \shownote      #1{#1}          \fi
\ifx \showarticletitle \undefined \def \showarticletitle #1{#1}   \fi
\ifx \showURL      \undefined \def \showURL       {\relax}        \fi
% The following commands are used for tagged output and should be
% invisible to TeX
\providecommand\bibfield[2]{#2}
\providecommand\bibinfo[2]{#2}
\providecommand\natexlab[1]{#1}
\providecommand\showeprint[2][]{arXiv:#2}

\bibitem[Bi et~al\mbox{.}(2021)]%
        {bi2021leveraging}
\bibfield{author}{\bibinfo{person}{Shuxian Bi}, \bibinfo{person}{Chaozhuo Li}, \bibinfo{person}{Xiao Han}, \bibinfo{person}{Zheng Liu}, \bibinfo{person}{Xing Xie}, \bibinfo{person}{Haizhen Huang}, {and} \bibinfo{person}{Zengxuan Wen}.} \bibinfo{year}{2021}\natexlab{}.
\newblock \showarticletitle{Leveraging Bidding Graphs for Advertiser-Aware Relevance Modeling in Sponsored Search}. In \bibinfo{booktitle}{\emph{Findings of the Association for Computational Linguistics: EMNLP 2021}}. \bibinfo{pages}{2215--2224}.
\newblock


\bibitem[Bo et~al\mbox{.}(2021)]%
        {bo2021beyond}
\bibfield{author}{\bibinfo{person}{Deyu Bo}, \bibinfo{person}{Xiao Wang}, \bibinfo{person}{Chuan Shi}, {and} \bibinfo{person}{Huawei Shen}.} \bibinfo{year}{2021}\natexlab{}.
\newblock \showarticletitle{Beyond low-frequency information in graph convolutional networks}. In \bibinfo{booktitle}{\emph{Proceedings of the AAAI Conference on Artificial Intelligence}}, Vol.~\bibinfo{volume}{35}. \bibinfo{pages}{3950--3957}.
\newblock


\bibitem[Cai and Wang(2020)]%
        {cai2020note}
\bibfield{author}{\bibinfo{person}{Chen Cai} {and} \bibinfo{person}{Yusu Wang}.} \bibinfo{year}{2020}\natexlab{}.
\newblock \showarticletitle{A note on over-smoothing for graph neural networks}.
\newblock \bibinfo{journal}{\emph{arXiv preprint arXiv:2006.13318}} (\bibinfo{year}{2020}).
\newblock


\bibitem[Chen et~al\mbox{.}(2020c)]%
        {chen2020measuring}
\bibfield{author}{\bibinfo{person}{Deli Chen}, \bibinfo{person}{Yankai Lin}, \bibinfo{person}{Wei Li}, \bibinfo{person}{Peng Li}, \bibinfo{person}{Jie Zhou}, {and} \bibinfo{person}{Xu Sun}.} \bibinfo{year}{2020}\natexlab{c}.
\newblock \showarticletitle{Measuring and relieving the over-smoothing problem for graph neural networks from the topological view}. In \bibinfo{booktitle}{\emph{Proceedings of the AAAI Conference on Artificial Intelligence}}, Vol.~\bibinfo{volume}{34}. \bibinfo{pages}{3438--3445}.
\newblock


\bibitem[Chen et~al\mbox{.}(2020b)]%
        {chen2020simple}
\bibfield{author}{\bibinfo{person}{Ting Chen}, \bibinfo{person}{Simon Kornblith}, \bibinfo{person}{Mohammad Norouzi}, {and} \bibinfo{person}{Geoffrey Hinton}.} \bibinfo{year}{2020}\natexlab{b}.
\newblock \showarticletitle{A simple framework for contrastive learning of visual representations}. In \bibinfo{booktitle}{\emph{International conference on machine learning}}. PMLR, \bibinfo{pages}{1597--1607}.
\newblock


\bibitem[Chen et~al\mbox{.}(2020a)]%
        {chen2020improved}
\bibfield{author}{\bibinfo{person}{Xinlei Chen}, \bibinfo{person}{Haoqi Fan}, \bibinfo{person}{Ross Girshick}, {and} \bibinfo{person}{Kaiming He}.} \bibinfo{year}{2020}\natexlab{a}.
\newblock \showarticletitle{Improved baselines with momentum contrastive learning}.
\newblock \bibinfo{journal}{\emph{arXiv preprint arXiv:2003.04297}} (\bibinfo{year}{2020}).
\newblock


\bibitem[Chen et~al\mbox{.}(2019)]%
        {chen2019drop}
\bibfield{author}{\bibinfo{person}{Yunpeng Chen}, \bibinfo{person}{Haoqi Fan}, \bibinfo{person}{Bing Xu}, \bibinfo{person}{Zhicheng Yan}, \bibinfo{person}{Yannis Kalantidis}, \bibinfo{person}{Marcus Rohrbach}, \bibinfo{person}{Shuicheng Yan}, {and} \bibinfo{person}{Jiashi Feng}.} \bibinfo{year}{2019}\natexlab{}.
\newblock \showarticletitle{Drop an octave: Reducing spatial redundancy in convolutional neural networks with octave convolution}. In \bibinfo{booktitle}{\emph{Proceedings of the IEEE/CVF International Conference on Computer Vision}}. \bibinfo{pages}{3435--3444}.
\newblock


\bibitem[Chien et~al\mbox{.}(2021)]%
        {chien2021node}
\bibfield{author}{\bibinfo{person}{Eli Chien}, \bibinfo{person}{Wei-Cheng Chang}, \bibinfo{person}{Cho-Jui Hsieh}, \bibinfo{person}{Hsiang-Fu Yu}, \bibinfo{person}{Jiong Zhang}, \bibinfo{person}{Olgica Milenkovic}, {and} \bibinfo{person}{Inderjit~S Dhillon}.} \bibinfo{year}{2021}\natexlab{}.
\newblock \showarticletitle{Node feature extraction by self-supervised multi-scale neighborhood prediction}.
\newblock \bibinfo{journal}{\emph{arXiv preprint arXiv:2111.00064}} (\bibinfo{year}{2021}).
\newblock


\bibitem[Chung(1997)]%
        {chung1997spectral}
\bibfield{author}{\bibinfo{person}{Fan~RK Chung}.} \bibinfo{year}{1997}\natexlab{}.
\newblock \bibinfo{booktitle}{\emph{Spectral graph theory}}. Vol.~\bibinfo{volume}{92}.
\newblock \bibinfo{publisher}{American Mathematical Soc.}
\newblock


\bibitem[Devlin et~al\mbox{.}(2018)]%
        {devlin2018bert}
\bibfield{author}{\bibinfo{person}{Jacob Devlin}, \bibinfo{person}{Ming-Wei Chang}, \bibinfo{person}{Kenton Lee}, {and} \bibinfo{person}{Kristina Toutanova}.} \bibinfo{year}{2018}\natexlab{}.
\newblock \showarticletitle{Bert: Pre-training of deep bidirectional transformers for language understanding}.
\newblock \bibinfo{journal}{\emph{arXiv preprint arXiv:1810.04805}} (\bibinfo{year}{2018}).
\newblock


\bibitem[Eckart and Young(1936)]%
        {eckart1936approximation}
\bibfield{author}{\bibinfo{person}{Carl Eckart} {and} \bibinfo{person}{Gale Young}.} \bibinfo{year}{1936}\natexlab{}.
\newblock \showarticletitle{The approximation of one matrix by another of lower rank}.
\newblock \bibinfo{journal}{\emph{Psychometrika}} \bibinfo{volume}{1}, \bibinfo{number}{3} (\bibinfo{year}{1936}), \bibinfo{pages}{211--218}.
\newblock


\bibitem[Gururangan et~al\mbox{.}(2020)]%
        {gururangan2020don}
\bibfield{author}{\bibinfo{person}{Suchin Gururangan}, \bibinfo{person}{Ana Marasovi{\'c}}, \bibinfo{person}{Swabha Swayamdipta}, \bibinfo{person}{Kyle Lo}, \bibinfo{person}{Iz Beltagy}, \bibinfo{person}{Doug Downey}, {and} \bibinfo{person}{Noah~A Smith}.} \bibinfo{year}{2020}\natexlab{}.
\newblock \showarticletitle{Don't stop pretraining: Adapt language models to domains and tasks}.
\newblock \bibinfo{journal}{\emph{arXiv preprint arXiv:2004.10964}} (\bibinfo{year}{2020}).
\newblock


\bibitem[Hamilton et~al\mbox{.}(2017)]%
        {hamilton2017inductive}
\bibfield{author}{\bibinfo{person}{Will Hamilton}, \bibinfo{person}{Zhitao Ying}, {and} \bibinfo{person}{Jure Leskovec}.} \bibinfo{year}{2017}\natexlab{}.
\newblock \showarticletitle{Inductive representation learning on large graphs}.
\newblock \bibinfo{journal}{\emph{Advances in neural information processing systems}}  \bibinfo{volume}{30} (\bibinfo{year}{2017}).
\newblock


\bibitem[HaoChen et~al\mbox{.}(2021)]%
        {haochen2021provable}
\bibfield{author}{\bibinfo{person}{Jeff~Z HaoChen}, \bibinfo{person}{Colin Wei}, \bibinfo{person}{Adrien Gaidon}, {and} \bibinfo{person}{Tengyu Ma}.} \bibinfo{year}{2021}\natexlab{}.
\newblock \showarticletitle{Provable guarantees for self-supervised deep learning with spectral contrastive loss}.
\newblock \bibinfo{journal}{\emph{Advances in Neural Information Processing Systems}}  \bibinfo{volume}{34} (\bibinfo{year}{2021}), \bibinfo{pages}{5000--5011}.
\newblock


\bibitem[Hassani and Khasahmadi(2020)]%
        {hassani2020contrastive}
\bibfield{author}{\bibinfo{person}{Kaveh Hassani} {and} \bibinfo{person}{Amir~Hosein Khasahmadi}.} \bibinfo{year}{2020}\natexlab{}.
\newblock \showarticletitle{Contrastive multi-view representation learning on graphs}. In \bibinfo{booktitle}{\emph{International Conference on Machine Learning}}. PMLR, \bibinfo{pages}{4116--4126}.
\newblock


\bibitem[Hastie et~al\mbox{.}(2009)]%
        {hastie2009elements}
\bibfield{author}{\bibinfo{person}{Trevor Hastie}, \bibinfo{person}{Robert Tibshirani}, \bibinfo{person}{Jerome~H Friedman}, {and} \bibinfo{person}{Jerome~H Friedman}.} \bibinfo{year}{2009}\natexlab{}.
\newblock \bibinfo{booktitle}{\emph{The elements of statistical learning: data mining, inference, and prediction}}. Vol.~\bibinfo{volume}{2}.
\newblock \bibinfo{publisher}{Springer}.
\newblock


\bibitem[He et~al\mbox{.}(2020)]%
        {he2020momentum}
\bibfield{author}{\bibinfo{person}{Kaiming He}, \bibinfo{person}{Haoqi Fan}, \bibinfo{person}{Yuxin Wu}, \bibinfo{person}{Saining Xie}, {and} \bibinfo{person}{Ross Girshick}.} \bibinfo{year}{2020}\natexlab{}.
\newblock \showarticletitle{Momentum contrast for unsupervised visual representation learning}. In \bibinfo{booktitle}{\emph{Proceedings of the IEEE/CVF conference on computer vision and pattern recognition}}. \bibinfo{pages}{9729--9738}.
\newblock


\bibitem[He et~al\mbox{.}(2012)]%
        {he2012guided}
\bibfield{author}{\bibinfo{person}{Kaiming He}, \bibinfo{person}{Jian Sun}, {and} \bibinfo{person}{Xiaoou Tang}.} \bibinfo{year}{2012}\natexlab{}.
\newblock \showarticletitle{Guided image filtering}.
\newblock \bibinfo{journal}{\emph{IEEE transactions on pattern analysis and machine intelligence}} \bibinfo{volume}{35}, \bibinfo{number}{6} (\bibinfo{year}{2012}), \bibinfo{pages}{1397--1409}.
\newblock


\bibitem[Hu et~al\mbox{.}(2019)]%
        {hu2019strategies}
\bibfield{author}{\bibinfo{person}{Weihua Hu}, \bibinfo{person}{Bowen Liu}, \bibinfo{person}{Joseph Gomes}, \bibinfo{person}{Marinka Zitnik}, \bibinfo{person}{Percy Liang}, \bibinfo{person}{Vijay Pande}, {and} \bibinfo{person}{Jure Leskovec}.} \bibinfo{year}{2019}\natexlab{}.
\newblock \showarticletitle{Strategies for pre-training graph neural networks}.
\newblock \bibinfo{journal}{\emph{arXiv preprint arXiv:1905.12265}} (\bibinfo{year}{2019}).
\newblock


\bibitem[Jeh and Widom(2003)]%
        {jeh2003scaling}
\bibfield{author}{\bibinfo{person}{Glen Jeh} {and} \bibinfo{person}{Jennifer Widom}.} \bibinfo{year}{2003}\natexlab{}.
\newblock \showarticletitle{Scaling personalized web search}. In \bibinfo{booktitle}{\emph{Proceedings of the 12th international conference on World Wide Web}}. \bibinfo{pages}{271--279}.
\newblock


\bibitem[Jiao et~al\mbox{.}(2020)]%
        {jiao2020sub}
\bibfield{author}{\bibinfo{person}{Yizhu Jiao}, \bibinfo{person}{Yun Xiong}, \bibinfo{person}{Jiawei Zhang}, \bibinfo{person}{Yao Zhang}, \bibinfo{person}{Tianqi Zhang}, {and} \bibinfo{person}{Yangyong Zhu}.} \bibinfo{year}{2020}\natexlab{}.
\newblock \showarticletitle{Sub-graph contrast for scalable self-supervised graph representation learning}. In \bibinfo{booktitle}{\emph{2020 IEEE international conference on data mining (ICDM)}}. IEEE, \bibinfo{pages}{222--231}.
\newblock


\bibitem[Jin et~al\mbox{.}(2022c)]%
        {jin2022heterformer}
\bibfield{author}{\bibinfo{person}{Bowen Jin}, \bibinfo{person}{Yu Zhang}, \bibinfo{person}{Qi Zhu}, {and} \bibinfo{person}{Jiawei Han}.} \bibinfo{year}{2022}\natexlab{c}.
\newblock \showarticletitle{Heterformer: A Transformer Architecture for Node Representation Learning on Heterogeneous Text-Rich Networks}.
\newblock \bibinfo{journal}{\emph{arXiv preprint arXiv:2205.10282}} (\bibinfo{year}{2022}).
\newblock


\bibitem[Jin et~al\mbox{.}(2021)]%
        {jin2021bite}
\bibfield{author}{\bibinfo{person}{Di Jin}, \bibinfo{person}{Xiangchen Song}, \bibinfo{person}{Zhizhi Yu}, \bibinfo{person}{Ziyang Liu}, \bibinfo{person}{Heling Zhang}, \bibinfo{person}{Zhaomeng Cheng}, {and} \bibinfo{person}{Jiawei Han}.} \bibinfo{year}{2021}\natexlab{}.
\newblock \showarticletitle{Bite-gcn: A new GCN architecture via bidirectional convolution of topology and features on text-rich networks}. In \bibinfo{booktitle}{\emph{Proceedings of the 14th ACM International Conference on Web Search and Data Mining}}. \bibinfo{pages}{157--165}.
\newblock


\bibitem[Jin et~al\mbox{.}(2022a)]%
        {jin2022code}
\bibfield{author}{\bibinfo{person}{Yiqiao Jin}, \bibinfo{person}{Yunsheng Bai}, \bibinfo{person}{Yanqiao Zhu}, \bibinfo{person}{Yizhou Sun}, {and} \bibinfo{person}{Wei Wang}.} \bibinfo{year}{2022}\natexlab{a}.
\newblock \showarticletitle{Code Recommendation for Open Source Software Developers}. In \bibinfo{booktitle}{\emph{Proceedings of the ACM Web Conference 2023}}.
\newblock


\bibitem[Jin et~al\mbox{.}(2023)]%
        {jin2023predicting}
\bibfield{author}{\bibinfo{person}{Yiqiao Jin}, \bibinfo{person}{Yeon-Chang Lee}, \bibinfo{person}{Kartik Sharma}, \bibinfo{person}{Meng Ye}, \bibinfo{person}{Karan Sikka}, \bibinfo{person}{Ajay Divakaran}, {and} \bibinfo{person}{Srijan Kumar}.} \bibinfo{year}{2023}\natexlab{}.
\newblock \showarticletitle{Predicting Information Pathways Across Online Communities}. In \bibinfo{booktitle}{\emph{KDD}}.
\newblock


\bibitem[Jin et~al\mbox{.}(2022b)]%
        {jin2022towards}
\bibfield{author}{\bibinfo{person}{Yiqiao Jin}, \bibinfo{person}{Xiting Wang}, \bibinfo{person}{Ruichao Yang}, \bibinfo{person}{Yizhou Sun}, \bibinfo{person}{Wei Wang}, \bibinfo{person}{Hao Liao}, {and} \bibinfo{person}{Xing Xie}.} \bibinfo{year}{2022}\natexlab{b}.
\newblock \showarticletitle{Towards fine-grained reasoning for fake news detection}. In \bibinfo{booktitle}{\emph{Proceedings of the AAAI Conference on Artificial Intelligence}}, Vol.~\bibinfo{volume}{36}. \bibinfo{pages}{5746--5754}.
\newblock


\bibitem[Kim et~al\mbox{.}(2021)]%
        {kim2021self}
\bibfield{author}{\bibinfo{person}{Taeuk Kim}, \bibinfo{person}{Kang~Min Yoo}, {and} \bibinfo{person}{Sang-goo Lee}.} \bibinfo{year}{2021}\natexlab{}.
\newblock \showarticletitle{Self-guided contrastive learning for BERT sentence representations}.
\newblock \bibinfo{journal}{\emph{arXiv preprint arXiv:2106.07345}} (\bibinfo{year}{2021}).
\newblock


\bibitem[Kingma and Ba(2014)]%
        {kingma2014adam}
\bibfield{author}{\bibinfo{person}{Diederik~P Kingma} {and} \bibinfo{person}{Jimmy Ba}.} \bibinfo{year}{2014}\natexlab{}.
\newblock \showarticletitle{Adam: A method for stochastic optimization}.
\newblock \bibinfo{journal}{\emph{arXiv preprint arXiv:1412.6980}} (\bibinfo{year}{2014}).
\newblock


\bibitem[Kong et~al\mbox{.}(2019)]%
        {kong2019mutual}
\bibfield{author}{\bibinfo{person}{Lingpeng Kong}, \bibinfo{person}{Cyprien de~Masson d'Autume}, \bibinfo{person}{Wang Ling}, \bibinfo{person}{Lei Yu}, \bibinfo{person}{Zihang Dai}, {and} \bibinfo{person}{Dani Yogatama}.} \bibinfo{year}{2019}\natexlab{}.
\newblock \showarticletitle{A mutual information maximization perspective of language representation learning}.
\newblock \bibinfo{journal}{\emph{arXiv preprint arXiv:1910.08350}} (\bibinfo{year}{2019}).
\newblock


\bibitem[Lee et~al\mbox{.}(2021)]%
        {lee2021predicting}
\bibfield{author}{\bibinfo{person}{Jason~D Lee}, \bibinfo{person}{Qi Lei}, \bibinfo{person}{Nikunj Saunshi}, {and} \bibinfo{person}{Jiacheng Zhuo}.} \bibinfo{year}{2021}\natexlab{}.
\newblock \showarticletitle{Predicting what you already know helps: Provable self-supervised learning}.
\newblock \bibinfo{journal}{\emph{Advances in Neural Information Processing Systems}}  \bibinfo{volume}{34} (\bibinfo{year}{2021}), \bibinfo{pages}{309--323}.
\newblock


\bibitem[Li et~al\mbox{.}(2021)]%
        {li2021adsgnn}
\bibfield{author}{\bibinfo{person}{Chaozhuo Li}, \bibinfo{person}{Bochen Pang}, \bibinfo{person}{Yuming Liu}, \bibinfo{person}{Hao Sun}, \bibinfo{person}{Zheng Liu}, \bibinfo{person}{Xing Xie}, \bibinfo{person}{Tianqi Yang}, \bibinfo{person}{Yanling Cui}, \bibinfo{person}{Liangjie Zhang}, {and} \bibinfo{person}{Qi Zhang}.} \bibinfo{year}{2021}\natexlab{}.
\newblock \showarticletitle{Adsgnn: Behavior-graph augmented relevance modeling in sponsored search}. In \bibinfo{booktitle}{\emph{Proceedings of the 44th International ACM SIGIR Conference on Research and Development in Information Retrieval}}. \bibinfo{pages}{223--232}.
\newblock


\bibitem[Li et~al\mbox{.}(2019)]%
        {li2019adversarial}
\bibfield{author}{\bibinfo{person}{Chaozhuo Li}, \bibinfo{person}{Senzhang Wang}, \bibinfo{person}{Yukun Wang}, \bibinfo{person}{Philip Yu}, \bibinfo{person}{Yanbo Liang}, \bibinfo{person}{Yun Liu}, {and} \bibinfo{person}{Zhoujun Li}.} \bibinfo{year}{2019}\natexlab{}.
\newblock \showarticletitle{Adversarial learning for weakly-supervised social network alignment}. In \bibinfo{booktitle}{\emph{Proceedings of the AAAI conference on artificial intelligence}}, Vol.~\bibinfo{volume}{33}. \bibinfo{pages}{996--1003}.
\newblock


\bibitem[Li et~al\mbox{.}(2017)]%
        {li2017ppne}
\bibfield{author}{\bibinfo{person}{Chaozhuo Li}, \bibinfo{person}{Senzhang Wang}, \bibinfo{person}{Dejian Yang}, \bibinfo{person}{Zhoujun Li}, \bibinfo{person}{Yang Yang}, \bibinfo{person}{Xiaoming Zhang}, {and} \bibinfo{person}{Jianshe Zhou}.} \bibinfo{year}{2017}\natexlab{}.
\newblock \showarticletitle{PPNE: property preserving network embedding}. In \bibinfo{booktitle}{\emph{Database Systems for Advanced Applications: 22nd International Conference, DASFAA 2017, Suzhou, China, March 27-30, 2017, Proceedings, Part I 22}}. Springer, \bibinfo{pages}{163--179}.
\newblock


\bibitem[Li et~al\mbox{.}(2020)]%
        {li2020prototypical}
\bibfield{author}{\bibinfo{person}{Junnan Li}, \bibinfo{person}{Pan Zhou}, \bibinfo{person}{Caiming Xiong}, {and} \bibinfo{person}{Steven~CH Hoi}.} \bibinfo{year}{2020}\natexlab{}.
\newblock \showarticletitle{Prototypical contrastive learning of unsupervised representations}.
\newblock \bibinfo{journal}{\emph{arXiv preprint arXiv:2005.04966}} (\bibinfo{year}{2020}).
\newblock


\bibitem[Li et~al\mbox{.}(2018)]%
        {li2018deeper}
\bibfield{author}{\bibinfo{person}{Qimai Li}, \bibinfo{person}{Zhichao Han}, {and} \bibinfo{person}{Xiao-Ming Wu}.} \bibinfo{year}{2018}\natexlab{}.
\newblock \showarticletitle{Deeper insights into graph convolutional networks for semi-supervised learning}. In \bibinfo{booktitle}{\emph{Thirty-Second AAAI conference on artificial intelligence}}.
\newblock


\bibitem[Liu et~al\mbox{.}(2020)]%
        {liu2020towards}
\bibfield{author}{\bibinfo{person}{Meng Liu}, \bibinfo{person}{Hongyang Gao}, {and} \bibinfo{person}{Shuiwang Ji}.} \bibinfo{year}{2020}\natexlab{}.
\newblock \showarticletitle{Towards deeper graph neural networks}. In \bibinfo{booktitle}{\emph{Proceedings of the 26th ACM SIGKDD international conference on knowledge discovery \& data mining}}. \bibinfo{pages}{338--348}.
\newblock


\bibitem[Liu et~al\mbox{.}(2019)]%
        {liu2019fine}
\bibfield{author}{\bibinfo{person}{Zhenghao Liu}, \bibinfo{person}{Chenyan Xiong}, \bibinfo{person}{Maosong Sun}, {and} \bibinfo{person}{Zhiyuan Liu}.} \bibinfo{year}{2019}\natexlab{}.
\newblock \showarticletitle{Fine-grained fact verification with kernel graph attention network}.
\newblock \bibinfo{journal}{\emph{arXiv preprint arXiv:1910.09796}} (\bibinfo{year}{2019}).
\newblock


\bibitem[Long et~al\mbox{.}(2020)]%
        {long2020graph}
\bibfield{author}{\bibinfo{person}{Qingqing Long}, \bibinfo{person}{Yilun Jin}, \bibinfo{person}{Guojie Song}, \bibinfo{person}{Yi Li}, {and} \bibinfo{person}{Wei Lin}.} \bibinfo{year}{2020}\natexlab{}.
\newblock \showarticletitle{Graph structural-topic neural network}. In \bibinfo{booktitle}{\emph{Proceedings of the 26th ACM SIGKDD International Conference on Knowledge Discovery \& Data Mining}}. \bibinfo{pages}{1065--1073}.
\newblock


\bibitem[Long et~al\mbox{.}(2021)]%
        {long2021hgk}
\bibfield{author}{\bibinfo{person}{Qingqing Long}, \bibinfo{person}{Lingjun Xu}, \bibinfo{person}{Zheng Fang}, {and} \bibinfo{person}{Guojie Song}.} \bibinfo{year}{2021}\natexlab{}.
\newblock \showarticletitle{HGK-GNN: Heterogeneous Graph Kernel based Graph Neural Networks}. In \bibinfo{booktitle}{\emph{Proceedings of the 27th ACM SIGKDD Conference on Knowledge Discovery \& Data Mining}}. \bibinfo{pages}{1129--1138}.
\newblock


\bibitem[Lu et~al\mbox{.}(2020)]%
        {lu2020twinbert}
\bibfield{author}{\bibinfo{person}{Wenhao Lu}, \bibinfo{person}{Jian Jiao}, {and} \bibinfo{person}{Ruofei Zhang}.} \bibinfo{year}{2020}\natexlab{}.
\newblock \showarticletitle{Twinbert: Distilling knowledge to twin-structured compressed bert models for large-scale retrieval}. In \bibinfo{booktitle}{\emph{Proceedings of the 29th ACM International Conference on Information \& Knowledge Management}}. \bibinfo{pages}{2645--2652}.
\newblock


\bibitem[McAuley et~al\mbox{.}(2015)]%
        {mcauley2015image}
\bibfield{author}{\bibinfo{person}{Julian McAuley}, \bibinfo{person}{Christopher Targett}, \bibinfo{person}{Qinfeng Shi}, {and} \bibinfo{person}{Anton Van Den~Hengel}.} \bibinfo{year}{2015}\natexlab{}.
\newblock \showarticletitle{Image-based recommendations on styles and substitutes}. In \bibinfo{booktitle}{\emph{Proceedings of the 38th international ACM SIGIR conference on research and development in information retrieval}}. \bibinfo{pages}{43--52}.
\newblock


\bibitem[Oord et~al\mbox{.}(2018)]%
        {oord2018representation}
\bibfield{author}{\bibinfo{person}{Aaron van~den Oord}, \bibinfo{person}{Yazhe Li}, {and} \bibinfo{person}{Oriol Vinyals}.} \bibinfo{year}{2018}\natexlab{}.
\newblock \showarticletitle{Representation learning with contrastive predictive coding}.
\newblock \bibinfo{journal}{\emph{arXiv preprint arXiv:1807.03748}} (\bibinfo{year}{2018}).
\newblock


\bibitem[Pang et~al\mbox{.}(2022)]%
        {pang2022improving}
\bibfield{author}{\bibinfo{person}{Bochen Pang}, \bibinfo{person}{Chaozhuo Li}, \bibinfo{person}{Yuming Liu}, \bibinfo{person}{Jianxun Lian}, \bibinfo{person}{Jianan Zhao}, \bibinfo{person}{Hao Sun}, \bibinfo{person}{Weiwei Deng}, \bibinfo{person}{Xing Xie}, {and} \bibinfo{person}{Qi Zhang}.} \bibinfo{year}{2022}\natexlab{}.
\newblock \showarticletitle{Improving Relevance Modeling via Heterogeneous Behavior Graph Learning in Bing Ads}. In \bibinfo{booktitle}{\emph{Proceedings of the 28th ACM SIGKDD Conference on Knowledge Discovery and Data Mining}}. \bibinfo{pages}{3713--3721}.
\newblock


\bibitem[Peng et~al\mbox{.}(2020)]%
        {peng2020graph}
\bibfield{author}{\bibinfo{person}{Zhen Peng}, \bibinfo{person}{Wenbing Huang}, \bibinfo{person}{Minnan Luo}, \bibinfo{person}{Qinghua Zheng}, \bibinfo{person}{Yu Rong}, \bibinfo{person}{Tingyang Xu}, {and} \bibinfo{person}{Junzhou Huang}.} \bibinfo{year}{2020}\natexlab{}.
\newblock \showarticletitle{Graph representation learning via graphical mutual information maximization}. In \bibinfo{booktitle}{\emph{Proceedings of The Web Conference 2020}}. \bibinfo{pages}{259--270}.
\newblock


\bibitem[Saunshi et~al\mbox{.}(2022)]%
        {saunshi2022understanding}
\bibfield{author}{\bibinfo{person}{Nikunj Saunshi}, \bibinfo{person}{Jordan Ash}, \bibinfo{person}{Surbhi Goel}, \bibinfo{person}{Dipendra Misra}, \bibinfo{person}{Cyril Zhang}, \bibinfo{person}{Sanjeev Arora}, \bibinfo{person}{Sham Kakade}, {and} \bibinfo{person}{Akshay Krishnamurthy}.} \bibinfo{year}{2022}\natexlab{}.
\newblock \showarticletitle{Understanding contrastive learning requires incorporating inductive biases}. In \bibinfo{booktitle}{\emph{International Conference on Machine Learning}}. PMLR, \bibinfo{pages}{19250--19286}.
\newblock


\bibitem[Shuman et~al\mbox{.}(2013)]%
        {shuman2013emerging}
\bibfield{author}{\bibinfo{person}{David~I Shuman}, \bibinfo{person}{Sunil~K Narang}, \bibinfo{person}{Pascal Frossard}, \bibinfo{person}{Antonio Ortega}, {and} \bibinfo{person}{Pierre Vandergheynst}.} \bibinfo{year}{2013}\natexlab{}.
\newblock \showarticletitle{The emerging field of signal processing on graphs: Extending high-dimensional data analysis to networks and other irregular domains}.
\newblock \bibinfo{journal}{\emph{IEEE signal processing magazine}} \bibinfo{volume}{30}, \bibinfo{number}{3} (\bibinfo{year}{2013}), \bibinfo{pages}{83--98}.
\newblock


\bibitem[Tang et~al\mbox{.}(2008)]%
        {tang2008arnetminer}
\bibfield{author}{\bibinfo{person}{Jie Tang}, \bibinfo{person}{Jing Zhang}, \bibinfo{person}{Limin Yao}, \bibinfo{person}{Juanzi Li}, \bibinfo{person}{Li Zhang}, {and} \bibinfo{person}{Zhong Su}.} \bibinfo{year}{2008}\natexlab{}.
\newblock \showarticletitle{Arnetminer: extraction and mining of academic social networks}. In \bibinfo{booktitle}{\emph{Proceedings of the 14th ACM SIGKDD international conference on Knowledge discovery and data mining}}. \bibinfo{pages}{990--998}.
\newblock


\bibitem[Tian(2022)]%
        {tian2022deep}
\bibfield{author}{\bibinfo{person}{Yuandong Tian}.} \bibinfo{year}{2022}\natexlab{}.
\newblock \showarticletitle{Deep contrastive learning is provably (almost) principal component analysis}.
\newblock \bibinfo{journal}{\emph{arXiv preprint arXiv:2201.12680}} (\bibinfo{year}{2022}).
\newblock


\bibitem[Veli{\v{c}}kovi{\'c} et~al\mbox{.}(2017)]%
        {velivckovic2017graph}
\bibfield{author}{\bibinfo{person}{Petar Veli{\v{c}}kovi{\'c}}, \bibinfo{person}{Guillem Cucurull}, \bibinfo{person}{Arantxa Casanova}, \bibinfo{person}{Adriana Romero}, \bibinfo{person}{Pietro Lio}, {and} \bibinfo{person}{Yoshua Bengio}.} \bibinfo{year}{2017}\natexlab{}.
\newblock \showarticletitle{Graph attention networks}.
\newblock \bibinfo{journal}{\emph{arXiv preprint arXiv:1710.10903}} (\bibinfo{year}{2017}).
\newblock


\bibitem[Velickovic et~al\mbox{.}(2019)]%
        {velickovic2019deep}
\bibfield{author}{\bibinfo{person}{Petar Velickovic}, \bibinfo{person}{William Fedus}, \bibinfo{person}{William~L Hamilton}, \bibinfo{person}{Pietro Li{\`o}}, \bibinfo{person}{Yoshua Bengio}, {and} \bibinfo{person}{R~Devon Hjelm}.} \bibinfo{year}{2019}\natexlab{}.
\newblock \showarticletitle{Deep Graph Infomax.}
\newblock \bibinfo{journal}{\emph{ICLR (Poster)}} \bibinfo{volume}{2}, \bibinfo{number}{3} (\bibinfo{year}{2019}), \bibinfo{pages}{4}.
\newblock


\bibitem[Von~Luxburg(2007)]%
        {von2007tutorial}
\bibfield{author}{\bibinfo{person}{Ulrike Von~Luxburg}.} \bibinfo{year}{2007}\natexlab{}.
\newblock \showarticletitle{A tutorial on spectral clustering}.
\newblock \bibinfo{journal}{\emph{Statistics and computing}} \bibinfo{volume}{17}, \bibinfo{number}{4} (\bibinfo{year}{2007}), \bibinfo{pages}{395--416}.
\newblock


\bibitem[Wang et~al\mbox{.}(2016a)]%
        {wang2016text}
\bibfield{author}{\bibinfo{person}{Chenguang Wang}, \bibinfo{person}{Yangqiu Song}, \bibinfo{person}{Haoran Li}, \bibinfo{person}{Ming Zhang}, {and} \bibinfo{person}{Jiawei Han}.} \bibinfo{year}{2016}\natexlab{a}.
\newblock \showarticletitle{Text classification with heterogeneous information network kernels}. In \bibinfo{booktitle}{\emph{Proceedings of the AAAI Conference on Artificial Intelligence}}, Vol.~\bibinfo{volume}{30}.
\newblock


\bibitem[Wang et~al\mbox{.}(2016b)]%
        {wang2016linked}
\bibfield{author}{\bibinfo{person}{Suhang Wang}, \bibinfo{person}{Jiliang Tang}, \bibinfo{person}{Charu Aggarwal}, {and} \bibinfo{person}{Huan Liu}.} \bibinfo{year}{2016}\natexlab{b}.
\newblock \showarticletitle{Linked document embedding for classification}. In \bibinfo{booktitle}{\emph{Proceedings of the 25th ACM international on conference on information and knowledge management}}. \bibinfo{pages}{115--124}.
\newblock


\bibitem[Wang et~al\mbox{.}(2019b)]%
        {wang2019improving}
\bibfield{author}{\bibinfo{person}{Wenlin Wang}, \bibinfo{person}{Chenyang Tao}, \bibinfo{person}{Zhe Gan}, \bibinfo{person}{Guoyin Wang}, \bibinfo{person}{Liqun Chen}, \bibinfo{person}{Xinyuan Zhang}, \bibinfo{person}{Ruiyi Zhang}, \bibinfo{person}{Qian Yang}, \bibinfo{person}{Ricardo Henao}, {and} \bibinfo{person}{Lawrence Carin}.} \bibinfo{year}{2019}\natexlab{b}.
\newblock \showarticletitle{Improving textual network learning with variational homophilic embeddings}.
\newblock \bibinfo{journal}{\emph{Advances in Neural Information Processing Systems}}  \bibinfo{volume}{32} (\bibinfo{year}{2019}).
\newblock


\bibitem[Wang et~al\mbox{.}(2021a)]%
        {wang2021kepler}
\bibfield{author}{\bibinfo{person}{Xiaozhi Wang}, \bibinfo{person}{Tianyu Gao}, \bibinfo{person}{Zhaocheng Zhu}, \bibinfo{person}{Zhengyan Zhang}, \bibinfo{person}{Zhiyuan Liu}, \bibinfo{person}{Juanzi Li}, {and} \bibinfo{person}{Jian Tang}.} \bibinfo{year}{2021}\natexlab{a}.
\newblock \showarticletitle{KEPLER: A unified model for knowledge embedding and pre-trained language representation}.
\newblock \bibinfo{journal}{\emph{Transactions of the Association for Computational Linguistics}}  \bibinfo{volume}{9} (\bibinfo{year}{2021}), \bibinfo{pages}{176--194}.
\newblock


\bibitem[Wang et~al\mbox{.}(2019a)]%
        {wang2019heterogeneous}
\bibfield{author}{\bibinfo{person}{Xiao Wang}, \bibinfo{person}{Houye Ji}, \bibinfo{person}{Chuan Shi}, \bibinfo{person}{Bai Wang}, \bibinfo{person}{Yanfang Ye}, \bibinfo{person}{Peng Cui}, {and} \bibinfo{person}{Philip~S Yu}.} \bibinfo{year}{2019}\natexlab{a}.
\newblock \showarticletitle{Heterogeneous graph attention network}. In \bibinfo{booktitle}{\emph{The world wide web conference}}. \bibinfo{pages}{2022--2032}.
\newblock


\bibitem[Wang et~al\mbox{.}(2021b)]%
        {wang2021self}
\bibfield{author}{\bibinfo{person}{Xiao Wang}, \bibinfo{person}{Nian Liu}, \bibinfo{person}{Hui Han}, {and} \bibinfo{person}{Chuan Shi}.} \bibinfo{year}{2021}\natexlab{b}.
\newblock \showarticletitle{Self-supervised heterogeneous graph neural network with co-contrastive learning}. In \bibinfo{booktitle}{\emph{Proceedings of the 27th ACM SIGKDD conference on knowledge discovery \& data mining}}. \bibinfo{pages}{1726--1736}.
\newblock


\bibitem[Wang et~al\mbox{.}(2022)]%
        {wang2022adaptive}
\bibfield{author}{\bibinfo{person}{Yiqi Wang}, \bibinfo{person}{Chaozhuo Li}, \bibinfo{person}{Zheng Liu}, \bibinfo{person}{Mingzheng Li}, \bibinfo{person}{Jiliang Tang}, \bibinfo{person}{Xing Xie}, \bibinfo{person}{Lei Chen}, {and} \bibinfo{person}{Philip~S Yu}.} \bibinfo{year}{2022}\natexlab{}.
\newblock \showarticletitle{An adaptive graph pre-training framework for localized collaborative filtering}.
\newblock \bibinfo{journal}{\emph{ACM Transactions on Information Systems}} \bibinfo{volume}{41}, \bibinfo{number}{2} (\bibinfo{year}{2022}), \bibinfo{pages}{1--27}.
\newblock


\bibitem[Wen and Li(2021)]%
        {wen2021toward}
\bibfield{author}{\bibinfo{person}{Zixin Wen} {and} \bibinfo{person}{Yuanzhi Li}.} \bibinfo{year}{2021}\natexlab{}.
\newblock \showarticletitle{Toward understanding the feature learning process of self-supervised contrastive learning}. In \bibinfo{booktitle}{\emph{International Conference on Machine Learning}}. PMLR, \bibinfo{pages}{11112--11122}.
\newblock


\bibitem[Xia et~al\mbox{.}(2022)]%
        {xia2022progcl}
\bibfield{author}{\bibinfo{person}{Jun Xia}, \bibinfo{person}{Lirong Wu}, \bibinfo{person}{Ge Wang}, \bibinfo{person}{Jintao Chen}, {and} \bibinfo{person}{Stan~Z Li}.} \bibinfo{year}{2022}\natexlab{}.
\newblock \showarticletitle{Progcl: Rethinking hard negative mining in graph contrastive learning}. In \bibinfo{booktitle}{\emph{International Conference on Machine Learning}}. PMLR, \bibinfo{pages}{24332--24346}.
\newblock


\bibitem[Xu et~al\mbox{.}(2021)]%
        {xu2021self}
\bibfield{author}{\bibinfo{person}{Minghao Xu}, \bibinfo{person}{Hang Wang}, \bibinfo{person}{Bingbing Ni}, \bibinfo{person}{Hongyu Guo}, {and} \bibinfo{person}{Jian Tang}.} \bibinfo{year}{2021}\natexlab{}.
\newblock \showarticletitle{Self-supervised graph-level representation learning with local and global structure}. In \bibinfo{booktitle}{\emph{International Conference on Machine Learning}}. PMLR, \bibinfo{pages}{11548--11558}.
\newblock


\bibitem[Xu et~al\mbox{.}(2019)]%
        {xu2019deep}
\bibfield{author}{\bibinfo{person}{Zenan Xu}, \bibinfo{person}{Qinliang Su}, \bibinfo{person}{Xiaojun Quan}, {and} \bibinfo{person}{Weijia Zhang}.} \bibinfo{year}{2019}\natexlab{}.
\newblock \showarticletitle{A deep neural information fusion architecture for textual network embeddings}.
\newblock \bibinfo{journal}{\emph{arXiv preprint arXiv:1908.11057}} (\bibinfo{year}{2019}).
\newblock


\bibitem[Yan et~al\mbox{.}(2024)]%
        {yan2024comprehensive}
\bibfield{author}{\bibinfo{person}{Hao Yan}, \bibinfo{person}{Chaozhuo Li}, \bibinfo{person}{Ruosong Long}, \bibinfo{person}{Chao Yan}, \bibinfo{person}{Jianan Zhao}, \bibinfo{person}{Wenwen Zhuang}, \bibinfo{person}{Jun Yin}, \bibinfo{person}{Peiyan Zhang}, \bibinfo{person}{Weihao Han}, \bibinfo{person}{Hao Sun}, {et~al\mbox{.}}} \bibinfo{year}{2024}\natexlab{}.
\newblock \showarticletitle{A Comprehensive Study on Text-attributed Graphs: Benchmarking and Rethinking}.
\newblock \bibinfo{journal}{\emph{Advances in Neural Information Processing Systems}}  \bibinfo{volume}{36} (\bibinfo{year}{2024}).
\newblock


\bibitem[Yang et~al\mbox{.}(2015)]%
        {yang2015network}
\bibfield{author}{\bibinfo{person}{Cheng Yang}, \bibinfo{person}{Zhiyuan Liu}, \bibinfo{person}{Deli Zhao}, \bibinfo{person}{Maosong Sun}, {and} \bibinfo{person}{Edward~Y Chang}.} \bibinfo{year}{2015}\natexlab{}.
\newblock \showarticletitle{Network representation learning with rich text information.}. In \bibinfo{booktitle}{\emph{IJCAI}}, Vol.~\bibinfo{volume}{2015}. \bibinfo{pages}{2111--2117}.
\newblock


\bibitem[Yang et~al\mbox{.}(2021)]%
        {yang2021graphformers}
\bibfield{author}{\bibinfo{person}{Junhan Yang}, \bibinfo{person}{Zheng Liu}, \bibinfo{person}{Shitao Xiao}, \bibinfo{person}{Chaozhuo Li}, \bibinfo{person}{Defu Lian}, \bibinfo{person}{Sanjay Agrawal}, \bibinfo{person}{Amit Singh}, \bibinfo{person}{Guangzhong Sun}, {and} \bibinfo{person}{Xing Xie}.} \bibinfo{year}{2021}\natexlab{}.
\newblock \showarticletitle{GraphFormers: GNN-nested transformers for representation learning on textual graph}.
\newblock \bibinfo{journal}{\emph{Advances in Neural Information Processing Systems}}  \bibinfo{volume}{34} (\bibinfo{year}{2021}), \bibinfo{pages}{28798--28810}.
\newblock


\bibitem[Yang et~al\mbox{.}(2022)]%
        {yang2022reinforcement}
\bibfield{author}{\bibinfo{person}{Ruichao Yang}, \bibinfo{person}{Xiting Wang}, \bibinfo{person}{Yiqiao Jin}, \bibinfo{person}{Chaozhuo Li}, \bibinfo{person}{Jianxun Lian}, {and} \bibinfo{person}{Xing Xie}.} \bibinfo{year}{2022}\natexlab{}.
\newblock \showarticletitle{Reinforcement subgraph reasoning for fake news detection}. In \bibinfo{booktitle}{\emph{Proceedings of the 28th ACM SIGKDD Conference on Knowledge Discovery and Data Mining}}. \bibinfo{pages}{2253--2262}.
\newblock


\bibitem[Yasunaga et~al\mbox{.}(2022)]%
        {yasunaga2022linkbert}
\bibfield{author}{\bibinfo{person}{Michihiro Yasunaga}, \bibinfo{person}{Jure Leskovec}, {and} \bibinfo{person}{Percy Liang}.} \bibinfo{year}{2022}\natexlab{}.
\newblock \showarticletitle{LinkBERT: Pretraining Language Models with Document Links}. In \bibinfo{booktitle}{\emph{Proceedings of the 60th Annual Meeting of the Association for Computational Linguistics (Volume 1: Long Papers)}}. \bibinfo{pages}{8003--8016}.
\newblock


\bibitem[Yasunaga et~al\mbox{.}(2017)]%
        {yasunaga2017graph}
\bibfield{author}{\bibinfo{person}{Michihiro Yasunaga}, \bibinfo{person}{Rui Zhang}, \bibinfo{person}{Kshitijh Meelu}, \bibinfo{person}{Ayush Pareek}, \bibinfo{person}{Krishnan Srinivasan}, {and} \bibinfo{person}{Dragomir Radev}.} \bibinfo{year}{2017}\natexlab{}.
\newblock \showarticletitle{Graph-based neural multi-document summarization}.
\newblock \bibinfo{journal}{\emph{arXiv preprint arXiv:1706.06681}} (\bibinfo{year}{2017}).
\newblock


\bibitem[Zhang et~al\mbox{.}(2019a)]%
        {zhang2019heterogeneous}
\bibfield{author}{\bibinfo{person}{Chuxu Zhang}, \bibinfo{person}{Dongjin Song}, \bibinfo{person}{Chao Huang}, \bibinfo{person}{Ananthram Swami}, {and} \bibinfo{person}{Nitesh~V Chawla}.} \bibinfo{year}{2019}\natexlab{a}.
\newblock \showarticletitle{Heterogeneous graph neural network}. In \bibinfo{booktitle}{\emph{Proceedings of the 25th ACM SIGKDD international conference on knowledge discovery \& data mining}}. \bibinfo{pages}{793--803}.
\newblock


\bibitem[Zhang et~al\mbox{.}(2019b)]%
        {zhang2019shne}
\bibfield{author}{\bibinfo{person}{Chuxu Zhang}, \bibinfo{person}{Ananthram Swami}, {and} \bibinfo{person}{Nitesh~V Chawla}.} \bibinfo{year}{2019}\natexlab{b}.
\newblock \showarticletitle{Shne: Representation learning for semantic-associated heterogeneous networks}. In \bibinfo{booktitle}{\emph{Proceedings of the twelfth ACM international conference on web search and data mining}}. \bibinfo{pages}{690--698}.
\newblock


\bibitem[Zhang et~al\mbox{.}(2016)]%
        {zhang2016geoburst}
\bibfield{author}{\bibinfo{person}{Chao Zhang}, \bibinfo{person}{Guangyu Zhou}, \bibinfo{person}{Quan Yuan}, \bibinfo{person}{Honglei Zhuang}, \bibinfo{person}{Yu Zheng}, \bibinfo{person}{Lance Kaplan}, \bibinfo{person}{Shaowen Wang}, {and} \bibinfo{person}{Jiawei Han}.} \bibinfo{year}{2016}\natexlab{}.
\newblock \showarticletitle{Geoburst: Real-time local event detection in geo-tagged tweet streams}. In \bibinfo{booktitle}{\emph{Proceedings of the 39th International ACM SIGIR conference on Research and Development in Information Retrieval}}. \bibinfo{pages}{513--522}.
\newblock


\bibitem[Zhang et~al\mbox{.}(2017)]%
        {zhang2017mixup}
\bibfield{author}{\bibinfo{person}{Hongyi Zhang}, \bibinfo{person}{Moustapha Cisse}, \bibinfo{person}{Yann~N Dauphin}, {and} \bibinfo{person}{David Lopez-Paz}.} \bibinfo{year}{2017}\natexlab{}.
\newblock \showarticletitle{mixup: Beyond empirical risk minimization}.
\newblock \bibinfo{journal}{\emph{arXiv preprint arXiv:1710.09412}} (\bibinfo{year}{2017}).
\newblock


\bibitem[Zhang et~al\mbox{.}(2020)]%
        {zhang2020graph}
\bibfield{author}{\bibinfo{person}{Jiawei Zhang}, \bibinfo{person}{Haopeng Zhang}, \bibinfo{person}{Congying Xia}, {and} \bibinfo{person}{Li Sun}.} \bibinfo{year}{2020}\natexlab{}.
\newblock \showarticletitle{Graph-bert: Only attention is needed for learning graph representations}.
\newblock \bibinfo{journal}{\emph{arXiv preprint arXiv:2001.05140}} (\bibinfo{year}{2020}).
\newblock


\bibitem[Zhao et~al\mbox{.}(2022)]%
        {zhao2022learning}
\bibfield{author}{\bibinfo{person}{Jianan Zhao}, \bibinfo{person}{Meng Qu}, \bibinfo{person}{Chaozhuo Li}, \bibinfo{person}{Hao Yan}, \bibinfo{person}{Qian Liu}, \bibinfo{person}{Rui Li}, \bibinfo{person}{Xing Xie}, {and} \bibinfo{person}{Jian Tang}.} \bibinfo{year}{2022}\natexlab{}.
\newblock \showarticletitle{Learning on Large-scale Text-attributed Graphs via Variational Inference}.
\newblock \bibinfo{journal}{\emph{arXiv preprint arXiv:2210.14709}} (\bibinfo{year}{2022}).
\newblock


\bibitem[Zhao et~al\mbox{.}(2023)]%
        {zhao2023beyond}
\bibfield{author}{\bibinfo{person}{Yi Zhao}, \bibinfo{person}{Chaozhuo Li}, \bibinfo{person}{Jiquan Peng}, \bibinfo{person}{Xiaohan Fang}, \bibinfo{person}{Feiran Huang}, \bibinfo{person}{Senzhang Wang}, \bibinfo{person}{Xing Xie}, {and} \bibinfo{person}{Jibing Gong}.} \bibinfo{year}{2023}\natexlab{}.
\newblock \showarticletitle{Beyond the Overlapping Users: Cross-Domain Recommendation via Adaptive Anchor Link Learning}. In \bibinfo{booktitle}{\emph{Proceedings of the 46th International ACM SIGIR Conference on Research and Development in Information Retrieval}}. \bibinfo{pages}{1488--1497}.
\newblock


\bibitem[Zhou et~al\mbox{.}(2020)]%
        {zhou2020s3}
\bibfield{author}{\bibinfo{person}{Kun Zhou}, \bibinfo{person}{Hui Wang}, \bibinfo{person}{Wayne~Xin Zhao}, \bibinfo{person}{Yutao Zhu}, \bibinfo{person}{Sirui Wang}, \bibinfo{person}{Fuzheng Zhang}, \bibinfo{person}{Zhongyuan Wang}, {and} \bibinfo{person}{Ji-Rong Wen}.} \bibinfo{year}{2020}\natexlab{}.
\newblock \showarticletitle{S3-rec: Self-supervised learning for sequential recommendation with mutual information maximization}. In \bibinfo{booktitle}{\emph{Proceedings of the 29th ACM international conference on information \& knowledge management}}. \bibinfo{pages}{1893--1902}.
\newblock


\bibitem[Zhu et~al\mbox{.}(2021)]%
        {zhu2021textgnn}
\bibfield{author}{\bibinfo{person}{Jason Zhu}, \bibinfo{person}{Yanling Cui}, \bibinfo{person}{Yuming Liu}, \bibinfo{person}{Hao Sun}, \bibinfo{person}{Xue Li}, \bibinfo{person}{Markus Pelger}, \bibinfo{person}{Tianqi Yang}, \bibinfo{person}{Liangjie Zhang}, \bibinfo{person}{Ruofei Zhang}, {and} \bibinfo{person}{Huasha Zhao}.} \bibinfo{year}{2021}\natexlab{}.
\newblock \showarticletitle{Textgnn: Improving text encoder via graph neural network in sponsored search}. In \bibinfo{booktitle}{\emph{Proceedings of the Web Conference 2021}}. \bibinfo{pages}{2848--2857}.
\newblock


\end{thebibliography}
% \clearpage
\appendix
\section{Theoretical Analysis of HFC}
\subsection{Background: Spectral Clustering}
\label{sec:spectral}
Given a graph $\mathcal{G}=(\mathcal{V},\mathcal{E})$, with adjacency matrix $A$, the Laplacian matrix of the graph is defined as $L=D - A$, where $D = diag(d_{1},...,d_{N})$ is the diagonal degree matrix ($d_{i}=\Sigma_{j}A_{i,j}$). Then the symmetric normalized Laplacian matrix is defined as $L_{sym}=D^{-\frac{1}{2}}LD^{-\frac{1}{2}}$. As $L_{sym}$ is real symmetric and positive semidefinite, therefore it can be diagonalized as $L=U\Lambda U^{T}$~\citep{chung1997spectral}. Here $U\in \mathbb{R}^{N\times N}=[u_{1},...,u_{N}]$, where $u_{i}\in \mathbb{R}^{N}$ denotes the $i$-th eigenvector of $L_{sym}$ and $\Lambda=diag(\lambda_{1},...,\lambda_{N})$ is the corresponding eigenvalue matrix. To partition the graph, spectral clustering~\citep{hastie2009elements,von2007tutorial} computes the first K eigenvectors and creates a feature vector $f_{K,v}\in \mathbb{R}^{K}$ for each node $v: \forall k \in [1,K], f_{K,v}(k)=u_{k}(v)$, which is in turn used to obtain K clusters by K-means or hierarchical clustering, etc.

An analogy between signals on graphs and usual signals~\citep{shuman2013emerging} suggests to interpret the spectrum of $L_{sym}$ as a Fourier domain for graphs, hence defining filters on graphs as diagonal operators after change of basis with $U^{-1}$. It turns out that the features $f_{K,v}$ can be
obtained by ideal low-pass filtering of the Delta function $\delta_{a}$ (localized at node a). Indeed, let $l_{K}$ be the step function where $l_{K}(\lambda)=1$ if $\lambda <\lambda_{K}$ and 0 otherwise. We define $L_{K}$ the diagonal matrix for which $L_{K}(i,i)=l_{K}(\lambda_{i})$. Then we have:
$f_{K,v}=L_{K}U^{-1}\delta_{v}\in \mathbb{R}^{K}$, where we fill the last $N-K$ values with 0’s. Therefore, spectral clustering is equivalent to clustering using low-pass filtering of the local descriptors $\delta_{v}$ of each node $v$ of the graph $\mathcal{G}$.

\subsection{Spectral Contrastive Loss Revisited}
\label{sec:scon}
To introduce spectral contrastive loss~\citep{haochen2021provable}, we give the definition of population view graph~\citep{haochen2021provable} first.

\noindent \textbf{Population View Graph.} A population view graph is defined as $\mathcal{G}=(\mathcal{X},\mathcal{W})$, where the set of nodes comprises all augmented views $\mathcal{X}$ of the population distribution, with $w_{xx^{'}}\in \mathcal{W}$ the edge weights of the edges connecting nodes $x,x^{'}$ that correspond to different views of the same input datapoint. The core assumption made is that this graph cannot be split into a large number of disconnected subgraphs. This set-up aligns well with the intuition that in order to generalize, the contrastive notion of “similarity” must extent beyond the purely single-instance-level, and must somehow connect distinct inputs points.

\noindent \textbf{Spectral Contrastive Loss.} Using the concept of population view graph, spectral contrastive loss is defined as:
\begin{equation}
\begin{aligned}
    \mathcal{L}(x,x^{+},x^{-},f_{\theta}) &= -2\cdot \mathbb{E}_{x,x^{+}}[f_{\theta}(x)^{T}f_{\theta}(x^{+})]\\
    &\ +\mathbb{E}_{x,x^{-}}[(f_{\theta}(x)^{T}f_{\theta}(x^{-}))^{2}],
\end{aligned}
\end{equation}
where $(x,x^{+})$ is a pair of views of the same datapoint, $(x,x^{-})$ is a pair of independently random views, and $f_{\theta}$ is a parameterized function from the data to $\mathbb{R}^{k}$. Minimizing spectral contrastive loss is equivalent to spectral clustering on the population view graph, where the top smallest eigenvectors of the Laplacian matrix are preserved as the columns of the final embedding matrix $F$. 

\subsection{HFC-aware Spectral Contrastive Loss}
As discussed in Appendix~\ref{sec:scon}, the spectral contrastive loss only learns the low-frequency component (LFC) of the graph from a spectral perspective, where the effects of high-frequency components (HFC) are much more attenuated. Recent studies have indicated that the LFC does not necessarily contain the most crucial information; while HFC may also encode useful information that is beneficial for the performance~\citep{bo2021beyond,chen2019drop}. 
% Moreover, LFC eventually leads to the over-smoothing problem~\citep{cai2020note,chen2020measuring,li2018deeper,liu2020towards} especially when the network exhibits a heterogeneous characteristic, where node representations converge to similar values, thus nodes cannot be easily distinguished. 
In this regard, merely using the spectral contrastive loss cannot adequately capture the varying significance of different frequency components, thus constraining the expressiveness of learned representations and producing suboptimal learning performance. How to incorporate the HFC to learn a more discriminative embedding still requires explorations.

In image signal processing, the Laplacian kernel is widely used to capture high-frequency edge information~\citep{he2012guided}. As its counterpart in
Graph Signal Processing (GSP)~\citep{shuman2013emerging}, we can multiply the graph Laplacian matrix $L$ with the input graph signal $x\in \mathbb{R}^{N}$, (\textit{i.e.,} $h=Lx$) to characterize its high-frequency components, which carry sharply varying signal information across edges of graph. On the contrary, when highlighting the LFC, we would subtract the term $Lx$ which emphasizes more on HFC from the input signal $x$, \textit{i.e.,} $z=x - Lx$. 

% It should be noted that the above operation corresponds to a fixed low-pass filter in the spectral domain, where higher weights are specified for LFC. 
% % However, in practice, HFC can also provide complementary insights for learning~\citep{bo2021beyond,chen2019drop}, especially when the label information is not smooth across edges. Additionally, the HFC of the input graph signal would be unavoidably too much weakened with fixed filters, leading to the well-known over-smoothing problem~\citep{li2018deeper}. 
% As discussed in Appendix~\ref{sec:spectral}, spectral clustering is equivalent to clustering using a low-pass filter on each node of the graph. Henceforth, the feature vectors learned by the spectral contrastive loss is LFC of the population view graph. In this regard, the fixed low-pass filters largely limit the fitting capability of contrastive learning and its variants for learning discriminative node representations. As a consequence, it is vital to capture the varying importance of frequencies in the filter to preserve more useful information and alleviate over-smoothing issues. 

As an alternative of the traditional low-pass filter, a simple and elegant solution to introduce HFC is to assign a single parameter to control the rate of high-frequency substraction. 
\begin{equation}
    z=x-\alpha Lx=(I-\alpha L)x\nonumber,
\end{equation}
where $I$ is the identity matrix. We thus obtain the kernel $I-\alpha L$ that contains HFC.

Following~\citep{haochen2021provable}, we consider the following matrix factorization based objective for eigenvectors:
\begin{equation}
\begin{aligned}
    \min_{F\in \mathbb{R}^{N\times K}}\mathcal{L}_{mf}(F) &=||(I-\alpha L)-FF^{T}||_{F}^{2}\\
    &=((1-\alpha)I+\Sigma_{i,j}(\frac{\alpha w_{x_{i},x_{j}}}{\sqrt{w_{x_{i}}}\sqrt{w_{x_{j}}}}-f_{\theta}(x_{i})^{T}f_{\theta}(x_{j})))^{2},
    \end{aligned}
\end{equation}
where $w_{x}=\Sigma_{x^{'}\in \mathcal{X}}w_{xx^{'}}$ is the total weights associated to view $x$. By the classical low-rank approximation theory (Eckart-Young-Mirsky theorem~\citep{eckart1936approximation}), minimizer $F$ possesses eigenvectors of HFC-aware kernel $I-\alpha L$ as columns and thus contains both the LFC and HFC of the population view graph.

\begin{lemma}
(HFC-aware spectral contrastive loss.) Denote $p_{x}$ is the $x$-th row of $F$. Let $p_{x}=w_{x}^{1/2}f_{\theta}(x)$. Then, the loss function $\mathcal{L}_{mf}(F)$ is equivalent to the following loss function for $f_{\theta}$, called HFC-aware spectral contrastive loss, up to an additive constant:
\begin{equation}
     \mathcal{L}_{mf}(F)=\mathcal{L}_{HFC}(f_{\theta})+const\nonumber, 
\end{equation}
where 
\begin{equation}
\begin{aligned}
    \mathcal{L}_{HFC}(f_{\theta}) &= -2\alpha \mathbb{E}_{x,x^{+}}[f_{\theta}(x)^{T}f_{\theta}(x^{+})]\\
    &\ \ \ \ +\mathbb{E}_{x,x^{-}}[(f_{\theta}(x)^{T}f_{\theta}(x^{-}))^{2}]
    \end{aligned}
\end{equation}

\end{lemma}

\begin{proof}
We expand $\mathcal{L}_{mf}(F)$ and obtain

\begin{equation}
\label{eq:all}
\begin{aligned}
    \mathcal{L}_{mf}(F)
    &=((1-\alpha)I+\Sigma_{i,j}(\frac{\alpha w_{x_{i},x_{j}}}{\sqrt{w_{x_{i}}}\sqrt{w_{x_{j}}}}-f_{\theta}(x_{i})^{T}f_{\theta}(x_{j})))^{2}   \\
    &=const - 2\Sigma_{i,j}[(1-\alpha)I+\frac{\alpha w_{x_{i},x_{j}}}{\sqrt{w_{x_{i}}}\sqrt{w_{x_{j}}}}]f_{\theta}(x_{i})^{T}f_{\theta}(x_{j})\\
    &\ \ \ \ +\Sigma_{i,j}(f_{\theta}(x_{i})^{T}f_{\theta}(x_{j}))^{2} \\
    &= 
    \left\{  
\begin{array}{l}
    const - 2\Sigma_{i,j}1-\alpha+\frac{\alpha w_{x_{i},x_{j}}}{\sqrt{w_{x_{i}}}\sqrt{w_{x_{j}}}}f_{\theta}(x_{i})^{T}f_{\theta}(x_{j})\\
    \ \ +\Sigma_{i,j}(f_{\theta}(x_{i})^{T}f_{\theta}(x_{j}))^{2}, i=j \\
    const - 2\Sigma_{i,j}\frac{\alpha w_{x_{i},x_{j}}}{\sqrt{w_{x_{i}}}\sqrt{w_{x_{j}}}}f_{\theta}(x_{i})^{T}f_{\theta}(x_{j})\\
    \ \ +\Sigma_{i,j}(f_{\theta}(x_{i})^{T}f_{\theta}(x_{j}))^{2}, i\ne j 
\end{array}
\right.  
\end{aligned}
\end{equation}
In our case two views $x_{i}$ and $x_{j}$ are not the same. We thus only focus on the $i\ne j$ case. Ignoring the scaling factor which doesn’t affect linear probe error, we can hence rewrite
the sum of last two terms of in Equation~\ref{eq:all} as Equation~\ref{eq:hfc}.
\end{proof}

\section{Node Classification}
\label{sec:node_classification}
\textbf{Settings.} In node classification, we train a 2-layer MLP classifier to classify nodes based on the output node representation embeddings of each method. The experiment is conducted on DBLP. Following~\citep{jin2022heterformer}, we select the most frequent 30 classes in DBLP. Also, we study both transductive and inductive node classification to understand the capability of our model comprehensively. For transductive node classification, the model has seen the classified nodes during representation learning (using the link prediction objective), while for inductive node classification, the model needs to predict the label of nodes not seen before. We separate the whole dataset into train set, validation set, and test set in 7:1:2 in all cases and each experiment is repeated 5 times in this section with the average performance reported.

\noindent \textbf{Results.} Table~\ref{tab:classification} demonstrates the results of different methods in transductive and inductive node classification. We observe that: (a) our HASH-CODE outperforms all the baseline methods significantly on both tasks, showing that HASH-CODE can learn more effective node representations for these tasks; (b) GNN-nested transformers generally achieve better results than GNN-cascaded transformers, which demonstrates the necessity of introducing graphic patterns in modeling textual representations; (c) HASH-CODE generalizes quite well on unseen nodes as its performance on inductive node classification is quite close to that on transductive node classification. Moreover, HASH-CODE even achieves higher performance in inductive settings than the baselines do in transductive settings.

\begin{table}[h]
% \Large
  \caption{Experiment results of transductive and inductive node classification on DBLP dataset. (HASH-CODE marked in bold, the best baseline underlined). HASH-CODE outperforms all baselines, especially the ones based on GNN-nested transformers.}
  \label{tab:classification}
  {
  \begin{tabular}{c|cc|cc}
    \toprule
    \multirow{2}{*}{Model} & \multicolumn{2}{c}{Transductive} & \multicolumn{2}{c}{Inductive}  \\
               ~&P@1 & NDCG  &P@1 & NDCG \\
    \midrule
    MeanSAGE & $0.5186$ & $0.7231$ & $0.5152$ & $0.7197$  \\
    GAT & $0.5208$ & $0.7196$ & $0.5126$ & $0.7146$  \\
    Bert & $0.5493$ & $0.7506$ & $0.5310$ & $0.7485$  \\
    Twin-Bert & $0.5291$ & $0.7440$ & $0.5248$ & $0.7431$  \\
    \midrule
    Bert+MeanSAGE & $0.6731$ & $0.7637$ & $0.6413$ & $0.7494$   \\
    Bert+MaxSAGE & $0.6705$ & $0.7752$ & $0.6587$ & $0.7599$   \\
    Bert+GAT & $0.6849$ & $0.0.7801$ & $0.6689$ & $0.0.7619$   \\
    \midrule
    TextGNN & $0.6820$ & $0.7753$ & $0.6380$ & $0.7716$   \\
    AdsGNN & $0.6882$ & $0.7790$ & $0.6624$ & $0.7737$   \\
    GraphFormers & $0.6919$ & $0.7929$ & $\underline{0.6791}$ & $0.7993$   \\
    Heterformer & $\underline{0.6924}$ & $\underline{0.7957}$ & $0.6746$ & $\underline{0.8079}$   \\
    \midrule
    HASH-CODE & $\textbf{0.7116}$ & $\textbf{0.8198}$ & $\textbf{0.6961}$ &  $\textbf{0.8170}$    \\
    \midrule
    \textit{Improv.} & $2.77\%$ & $3.03\%$ & $2.50\%$ &  $1.13\%$    \\
    \bottomrule
  \end{tabular}
  }
  % \vspace{-0.3cm}
\end{table}

\section{In-depth Analysis}
\subsection{Data Sparsity Analysis}
\label{sec:sparsity}
% Conventional representation learning methods require a considerable amount of training data, thus they are likely to suffer from the data sparsity issues in real-world applications. This problem can be alleviated by our method because the proposed contrastive learning approach can better utilize the data correlation from input. 
We
simulate the data sparsity scenarios by using different proportions of the full dataset.
Figure~\ref{fig:sparsity} shows the evaluation results on Product and Sports datasets. As
we can see, the performance substantially drops when less training data is used. While, HASH-CODE is consistently better than baselines in all cases, especially in an extreme sparsity level (20\%). This observation implies that HASH-CODE is able to make better use of the data with the contrastive learning method, which alleviates the influence of data
sparsity problem for representation learning to some extent.

\begin{figure}[h]
\subfigure[DBLP]{\centering
    \includegraphics[width=0.48\linewidth]{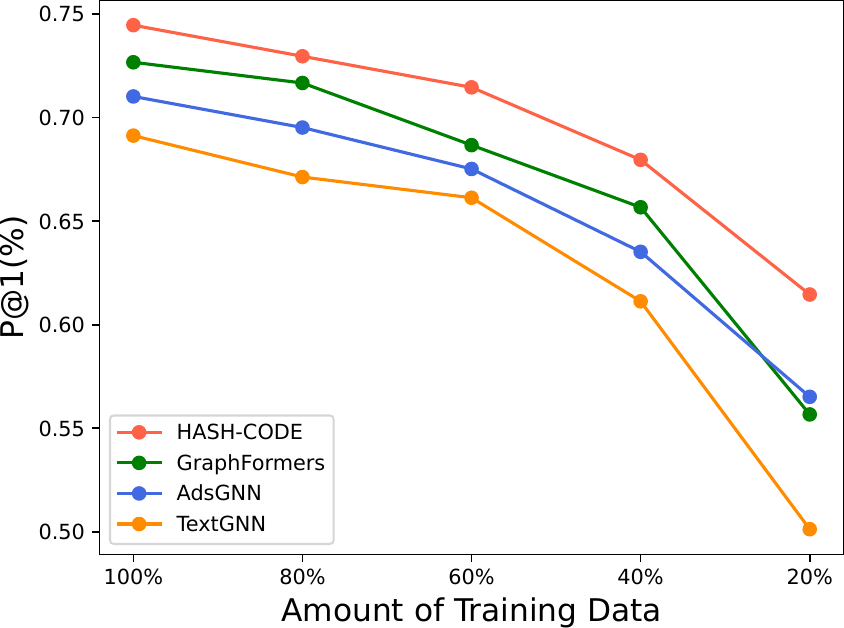}
    }
\subfigure[Product]{\centering
    \includegraphics[width=0.48\linewidth]{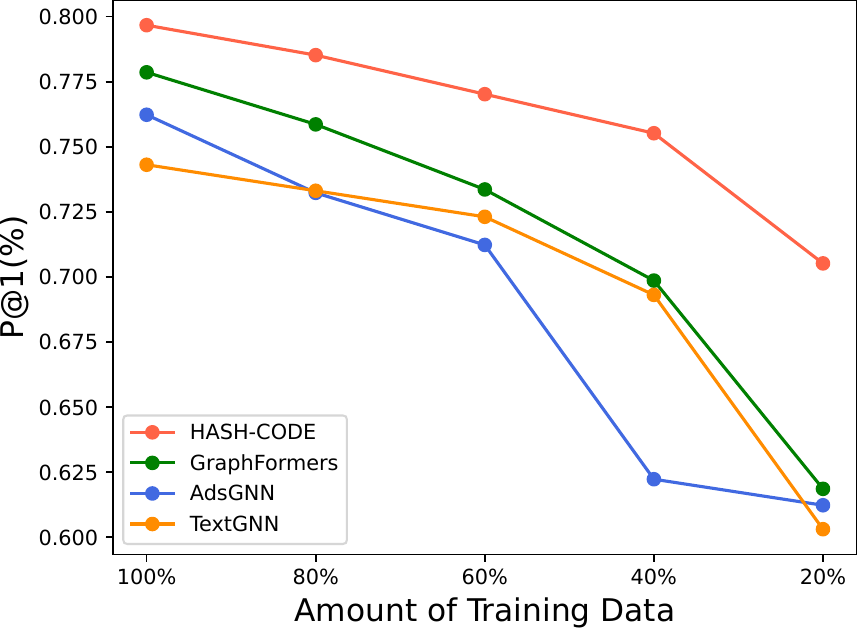}
    }
    \caption{Performance (P@1) comparison w.r.t. different sparsity levels on DBLP and Product datasets. The performance substantially drops when less training data is used, while  HASH-CODE is consistently better than baselines in all
cases, especially in an extreme sparsity level (20\%).}
    \label{fig:sparsity}
    % \vspace{-0.3cm}
\end{figure}

% \begin{figure*}[h]
% \centering
% \subfigure[Product]{
% \includegraphics[width=0.25\textwidth]{figures/Product_Data_Amount.pdf}
% \label{figure:prodcut_sparsity}
% }
% \subfigure[Beauty]{
% \includegraphics[width=0.25\textwidth]{figures/Beauty_Data_Amount.pdf}
% \label{figure:beauty_sparsity}
% }
% \subfigure[Sports]{
% \includegraphics[width=0.25\textwidth]{figures/Sports_Data_Amount.pdf}
% \label{figure:sports_sparsity}
% }

% \subfigure[Toys]{
% \includegraphics[width=0.25\textwidth]{figures/Toys_Data_Amount.pdf}
% \label{figure:toys_sparsity}
% }
% \subfigure[DBLP]{
% \includegraphics[width=0.25\textwidth]{figures/DBLP_Data_Amount.pdf}
% \label{figure:dblp_sparsity}
% }
% \subfigure[Wiki]{
% \includegraphics[width=0.25\textwidth]{figures/Wiki_Data_Amount.pdf}
% \label{figure:wiki_sparsity}
% }
% \caption{Performance (P@1) comparison w.r.t. different sparsity levels on DBLP and Product datasets. The performance substantially drops when less training data is used, while  HASH-CODE is consistently better than baselines in all
% cases, especially in an extreme sparsity level (20\%).}
% \label{fig:sparsity}
% \end{figure*}

\subsection{Influence of Training Epochs Number}
\label{sec:epochs}
% Our approach consists of co-training with GNNs and Transformers. During the training stage, our model can learn the enhanced representations of the attribute and node for the representation learning task. The number of training epochs will affect the performance of the downstream task. To
% investigate this, 

We train our model with a varying number of epochs and fine-tune it on the downstream task.
Figure~\ref{fig:epoch} presents the results on Product and Sports datasets. We can see that our model benefits mostly from the first 20 training epochs. And after that, the performance improves slightly. Based on this observation, we can conclude that the correlations among different views (i.e., the graph topology and textual attributes) can be well-captured by our contrastive learning
approach through training within a small number of epochs. So that the enhanced data representations can improve the performance of the downstream tasks.

\begin{figure}[h]
\subfigure[DBLP]{\centering
    \includegraphics[width=0.48\linewidth]{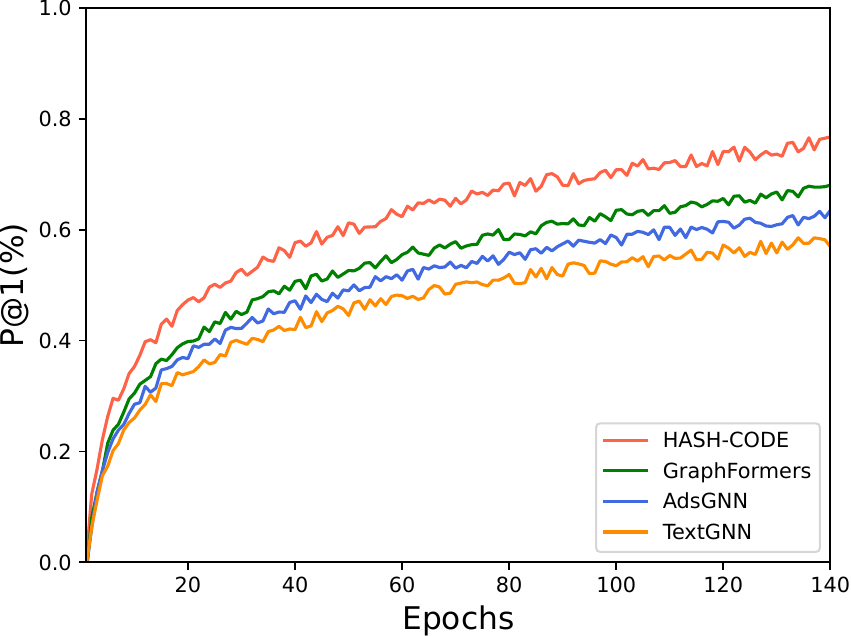}
    }
\subfigure[Product]{\centering
    \includegraphics[width=0.48\linewidth]{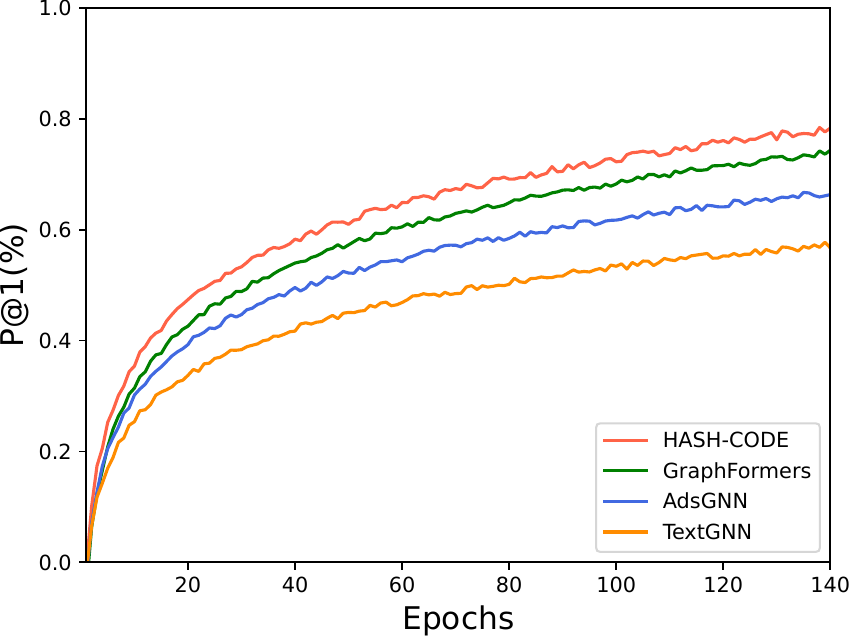}
    }
    \caption{Performance (P@1) comparison w.r.t. different numbers of training epochs on DBLP and Product datasets. HASH-CODE benefits mostly from the first 20 training epochs, thus  the correlations among different views can be well-captured by our approach through training within a small number of epochs.}
    \label{fig:epoch}
    % \vspace{-0.3cm}
\end{figure}

\begin{figure}[h]
\subfigure[DBLP]{\centering
    \includegraphics[width=0.48\linewidth]{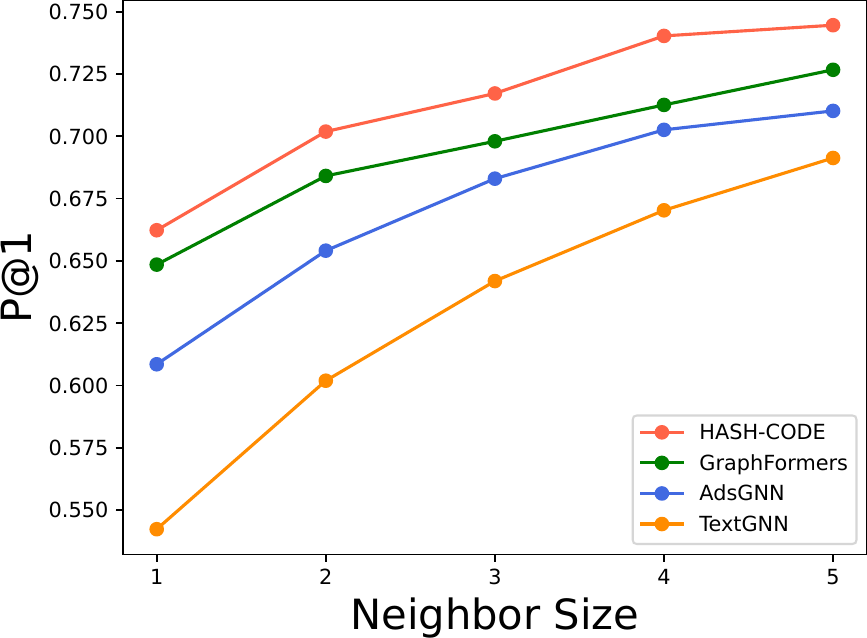}
    }
\subfigure[Product]{\centering
    \includegraphics[width=0.48\linewidth]{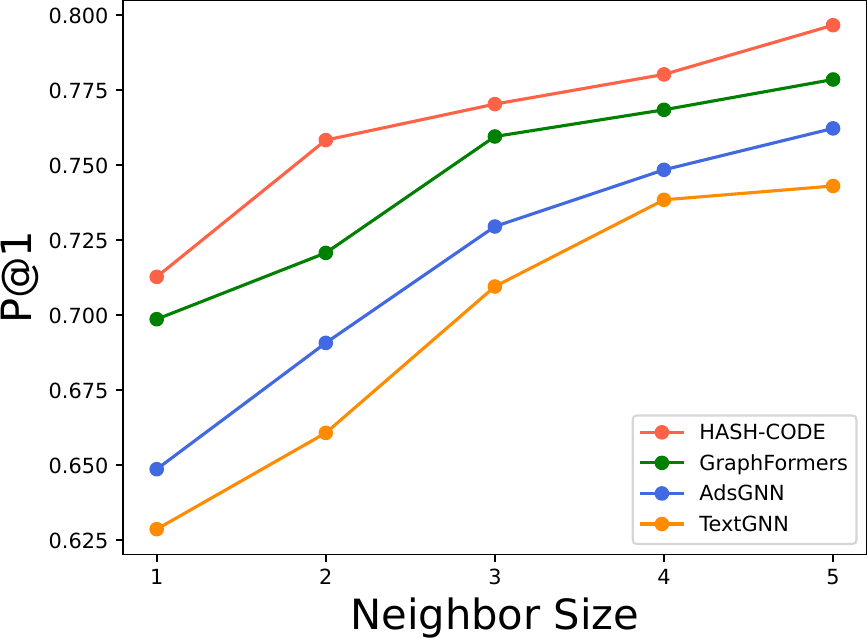}
    }
    \caption{Impact of neighbor size on DBLP and Product dataset. Enlarging the number of neighbour nodes brings performance improvement to both models. HASH-CODE maintains consistent advantages over GraphFormers over all test cases.}
    \label{fig:neibor}
    % \vspace{-0.3cm}
\end{figure}

% \begin{figure}[h]
% \centering
% \subfigure[Product]{
% \includegraphics[width=0.45\linewidth]{figures/Product_Size_P1.pdf}
% \label{figure:prodcut_size}
% }
% % \subfigure[Beauty]{
% % \includegraphics[width=0.3\textwidth]{figures/Beauty_Size_P1.pdf}
% % \label{figure:beauty_size}
% % }
% % \subfigure[Sports]{
% % \includegraphics[width=0.3\textwidth]{figures/Sports_Size_P1.pdf}
% % \label{figure:sports_size}
% % }

% % \subfigure[Toys]{
% % \includegraphics[width=0.3\textwidth]{figures/Toys_Size_P1.pdf}
% % \label{figure:toys_size}
% % }
% \subfigure[DBLP]{
% \includegraphics[width=0.45\linewidth]{figures/Neighbor_Size_P1.pdf}
% \label{figure:dblp_size}
% }
% % \subfigure[Wiki]{
% % \includegraphics[width=0.3\textwidth]{figures/Wiki_Size_P1.pdf}
% % \label{figure:wiki_size}
% % }
% \caption{Impact of neighbor size on DBLP dataset. Enlarging the number of neighbour nodes brings performance improvement to both models. HASH-CODE maintains consistent advantages over GraphFormers over all test cases.}
% \label{fig:neibor}
% \end{figure}

\subsection{Influence of Neighbor Size}
\label{sec:neighbor_size}
We analyze the impact of neighbourhood size with a fraction of neighbour nodes randomly sampled for each center node. From Figure~\ref{fig:neibor}, we can observe that with the increasing number of neighbour nodes, both HASH-CODE and Graphformers achieve higher prediction accuracies. However, the marginal gain is varnishing, as the relative improvement becomes smaller when more neighbours are included. In all the testing cases, HASH-CODE maintains consistent advantages over GraphFormers, which demonstrates the effectiveness of our proposed method.

% \subsection{HFC-aware Embedding Visualization.}
% \label{sec:visualization}
% To intuitively study the impact of our HFC-loss, we visualize the input node embeddings for different target classes by t-SNE ~\citep{van2008visualizing}. We conduct the visualization on DBLP with four different target classes, and each target class has more than 1000 node embeddings. Figure~\ref{fig:visulize} shows that compared with HFC-aware loss, the spectral contrastive loss cannot effectively distinguish different types of sample nodes. Especially in the central part of Figure~\ref{fig:nohfc}, sample points are almost completely overlapping. It is clear that the HFC-aware loss learns more discriminative node embeddings. 

% \begin{figure}[h]
% \centering
% \subfigure[HASH-CODE-NoHFC]{
% \begin{minipage}[t]{0.5\linewidth}
% \centering
% \includegraphics[width=\linewidth]{figures/HCL-NoHFC.png}
% \label{fig:nohfc}
% %\caption{fig1}
% \end{minipage}%
% }%
% \subfigure[HASH-CODE-HFC]{
% \begin{minipage}[t]{0.5\linewidth}
% \centering
% \includegraphics[width=\linewidth]{figures/HCL-HFC.png}
% \label{fig:hfc}
% %\caption{fig1}
% \end{minipage}%
% }%               
% \centering
% \caption{Embedding visulization of input nodes belonging to different target classes. Points with the same color denote input nodes belonging to the same target class. HFC-aware loss learns more discriminative embeddings than spectral contrastive loss.}
% \label{fig:visulize}
% % \vspace{-0.3cm}
% \end{figure}

%% The file named.bst is a bibliography style file for BibTeX 0.99c

\end{document}